\documentclass{llncs}
\usepackage{makeidx}  
\usepackage{graphicx,latexsym,amssymb,amsmath,alltt}
\usepackage{color}
\usepackage{algorithm}
\usepackage{algorithmic}

\definecolor{blank}{rgb}{0.7,0.7,0.7}

\long\def\comment#1{}

\renewcommand{\phi}{\varphi}

\def\defemb#1#2{\expandafter\def\csname #1\endcsname
                              {\relax\ifmmode #2\else\hbox{$#2$}\fi}}
\defemb{cA}{{\cal A}}
\defemb{cB}{{\cal B}}
\defemb{cC}{{\cal C}}
\defemb{cD}{{\cal D}}
\defemb{cE}{{\cal E}}
\defemb{cF}{{\cal F}}
\defemb{cG}{{\cal G}}
\defemb{cH}{{\cal H}}
\defemb{cI}{{\cal I}}
\defemb{cJ}{{\cal J}}
\defemb{cL}{{\cal L}}
\defemb{cM}{{\cal M}}
\defemb{cN}{{\cal N}}
\defemb{cO}{{\cal O}}
\defemb{cP}{{\cal P}}
\defemb{cR}{{\cal R}}
\defemb{cS}{{\cal S}}
\defemb{cT}{{\cal T}}
\defemb{cU}{{\cal U}}
\defemb{cV}{{\cal V}}
\defemb{cW}{{\cal W}}
\defemb{cX}{{\cal X}}
\defemb{cZ}{{\cal Z}}

\makeatletter
\newenvironment{prog}{\vspace{1.0ex}\par
\obeylines\@vobeyspaces\tt}{\vspace{1.0ex}\noindent
}
\makeatother
\newcommand{\startprog}{\begin{prog}}
\newcommand{\stopprog}{\end{prog}\noindent}

\catcode`\_=\active
\let_=\sb
\catcode`\_=12
\makeatother

\begin{document}
\frontmatter          
\pagestyle{headings}  
\addtocmark{Hamiltonian Mechanics} 


\title{Optimal Divide and Query\\ (extended version)\thanks{This work has been partially supported by the Spanish \emph{Ministerio de
Ciencia e Innovaci\'on} under grant TIN2008-06622-C03-02 and by the
\emph{Generalitat Valenciana} under grant PROMETEO/2011/052.}}
\titlerunning{Optimal Divide and Query} 
\author{David Insa \and Josep Silva}

\institute{
Universidad Polit\'ecnica de Valencia\\Camino de Vera s/n,
E-46022 Valencia, Spain. \\
\email{\{dinsa,jsilva\}@dsic.upv.es}
}

\maketitle              

\begin{abstract}
Algorithmic debugging is a semi-automatic debugging technique that allows the programmer to precisely identify the location of bugs without the need to inspect the source code. The technique has been successfully adapted to all paradigms and mature implementations have been released for languages such as Haskell, Prolog or Java. During three decades, the algorithm introduced by Shapiro and later improved by Hirunkitti 
has been thought optimal. In this paper we first show that this algorithm is not optimal, and moreover, in some situations it is unable to find all possible solutions, thus it is incomplete. Then, we present a new version of the algorithm that is proven optimal, and we introduce some equations that allow the algorithm to identify all optimal solutions.
\end{abstract}

\section{Introduction}

Debugging is one of the most important but less automated (and, thus, time-consuming) tasks in the software development process. 
The programmer is often forced to manually explore the code or iterate over it using, e.g., breakpoints, and this process usually requires a deep understanding of the source code to find the bug. \emph{Algorithmic debugging} \cite{Sha82} is a semi-automatic debugging technique that has been extended to practically all paradigms \cite{Sil07c}. Recent research has produced new advances to increase the scalability of the technique producing new scalable and mature debuggers. The technique is based on the answers of the programmer to a series of questions generated automatically by the algorithmic debugger.
The questions are always whether a given result of an activation of a subcomputation with given input values is actually correct.
The answers provide the debugger with information about the correctness of some (sub)computations of a given program; and the debugger uses them to guide the search for the bug until a buggy portion of code is isolated.

\begin{example}
\label{ex_motivating_example}
Consider this simple Haskell program inspired in a similar example by \cite{Fri92}.
 It wrongly (it has a bug) implements the sorting algorithm \emph{Insertion Sort}:

{\footnotesize
\begin{verbatim}
main = insort [2,1,3]

insort [] = []
insort (x:xs) = insert x (insort xs)

insert x [] = [x]
insert x (y:ys) = if x>=y then (x:y:ys)
                          else (y:(insert x ys))
\end{verbatim}
}

\noindent An algorithmic debugging session for this program is the following
({\tt YES} and {\tt NO} answers are provided by the programmer):

{\footnotesize
\begin{verbatim}
Starting Debugging Session...
(1)  insort [1,3] = [3,1]? NO
(2)  insort [3] = [3]? YES
(3)  insert 1 [3] = [3,1]? NO
(4)  insert 1 [] = [1]? YES

Bug found in rule:
insert x (y:ys) = if x>=y then _ else (y:(insert x ys))
\end{verbatim}
}
\noindent The debugger points out the part of the code that contains the bug.
In this case {\tt x>=y} should be {\tt x<=y}. Note that, to debug the program, the programmer only has to answer questions. It is not even necessary to see the code.
\end{example}

Typically, algorithmic debuggers have a front-end that produces a data structure representing a program execution---the so-called \textit{execution tree} (ET) \cite{nilphd98}---; and a back-end that uses the ET to ask questions and process the answers of the programmer to locate the bug. For instance, the ET of the program in Example~\ref{ex_motivating_example} is depicted in Figure~\ref{fig_ET_Insort}.

\begin{figure}[h]
\label{fig_ET_Insort}
  \centering
  \includegraphics[width=8cm]{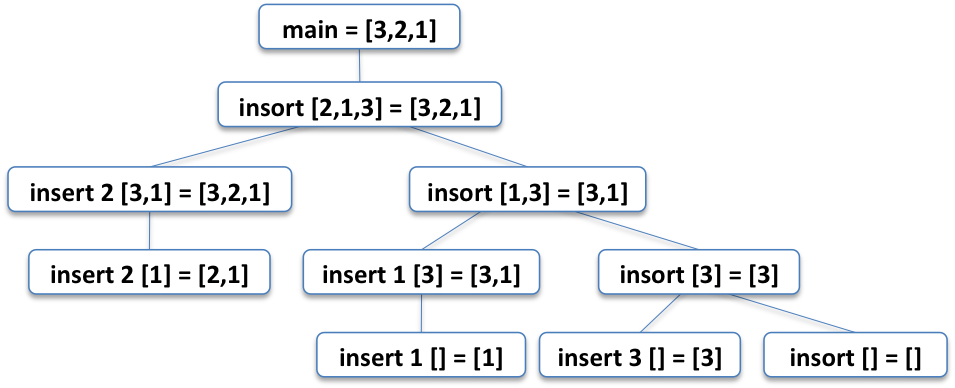}\\

\caption{ET of the program in Example~\ref{ex_motivating_example}}
\end{figure}

The strategy used to decide what nodes of the ET should be asked is crucial for the performance of the technique. Since the definition of algorithmic debugging, there has been a lot of research concerning the definition of new strategies trying to minimize the number of questions \cite{Sil07c}. We conducted several experiments to measure the performance of all current algorithmic debugging strategies. The results of the experiments are shown in Figure~\ref{fig_comparison}, where the first column contains the names of the benchmarks; column {\tt nodes} shows the number of nodes in the ET associated with each benchmark; and the other columns represent algorithmic debugging strategies \cite{Sil07c} that are ordered according to their performance: Optimal Divide \& Query ({\tt D\&QO}), Divide \& Query by Hirunkitti ({\tt D\&QH}), Divide \& Query by Shapiro ({\tt D\&QS}), Divide by Rules \& Query ({\tt DR\&Q}), Heaviest First ({\tt HF}), More Rules First ({\tt MRF}), Hat Delta Proportion ({\tt HD-P}), Top-Down ({\tt TD}), Hat Delta YES ({\tt HD-Y}), Hat Delta NO ({\tt HD-N}), Single Stepping ({\tt SS}).

\begin{figure*}
\centering
\includegraphics[width=12.3cm]{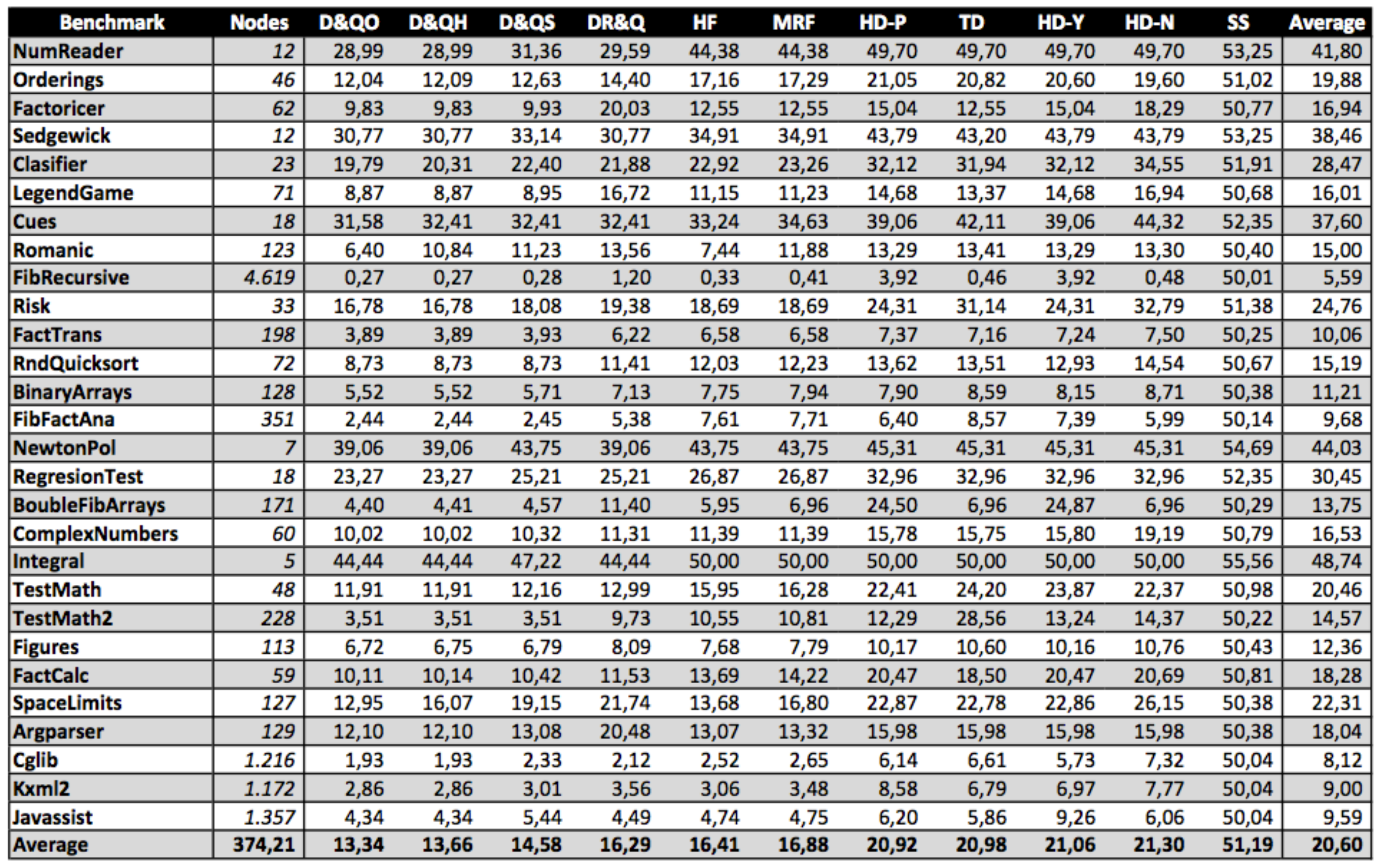}\\
\caption{Performance of algorithmic debugging strategies}\label{fig_comparison}
\end{figure*}

For each benchmark, we produced its associated ET and assumed that the buggy node could be any node of the ET (i.e., any subcomputation in the execution of the program could be buggy). Therefore, we performed a different experiment for each possible case and, hence, each cell of the table summarizes a number of experiments that were automatized. In particular, benchmark \emph{Factoricer} has been debugged 62 times with each strategy; each time we selected a different node and simulated that it was buggy, thus the results shown are the average number of questions performed by each strategy with respect to the number of nodes (i.e., the mean percentage of nodes asked). Similarly, benchmark \emph{Cglib} has been debugged 1216 times with each strategy, and so on.

Observe that the best algorithmic debugging strategies in practice are the two variants of Divide and Query (ignoring our new technique D\&QO). Moreover, from a theoretical point of view, this strategy has been thought optimal in the worst case for almost 30 years, and it has been implemented in almost all current algorithmic debuggers (see, e.g., \cite{Cab09b,Tho06,IS10,BerniePopePhd}). In this paper we show that current algorithms for D\&Q are suboptimal. We show the problems of D\&Q and solve them in a new improved algorithm that is proven optimal. Moreover, the original strategy was only defined for ETs where all the nodes have an individual weight of 1. 
In contrast, we allow our algorithms to work with different individual weights that can be integer, but also decimal. An individual weight of zero means that this node cannot contain the bug. A positive individual weight approximates the probability of being buggy. 
The higher the individual weight, the higher the probability. This generalization strongly influences the technique and allows us to assign different probabilities of being buggy to different parts of the program. For instance, a recursive function with higher-order calls should be assigned a higher individual weight than a function implementing a simple base case \cite{Sil07c}.
The weight of the nodes can also be reassigned dynamically during the debugging session in order to take into account the oracle's answers \cite{Tho06}. 

We show that the original algorithms are inefficient with ETs where nodes can have different individual weights in the domain of the positive real numbers (including zero) and we redefine the technique for these generalized ETs.

The rest of the paper has been organized as follows. In Section~\ref{sec_DQ} we recall and formalize the strategy D\&Q and we show with counterexamples that it is suboptimal and incomplete. Then, in Section~\ref{sec_optimal} we introduce two new algorithms for D\&Q that are optimal and complete. Each algorithm is useful for a different type of ET. 
Finally, Section~\ref{sec_conclusions} concludes. Proofs of technical results can be found in the appendix.

\section{D\&Q by Shapiro vs. D\&Q by Hirunkitti}
\label{sec_DQ}

In this section we formalize the strategy D\&Q to show the differences between the original version by Shapiro \cite{Sha82} and the improved version by Hirunkitti and Hogger \cite{Hir93}. We start with the definition of \emph{marked execution tree}, that is an ET where some nodes could have been removed because they were marked as correct (i.e., answered YES), some nodes could have been marked as wrong (i.e., answered NO) and the correctness of the other nodes is undefined.  

\begin{definition}[Marked Execution Tree]
A \emph{mar\-ked execution tree} (MET) is a tree $T = (N, E, M)$ where $N$ are the nodes, $E \subseteq N \times N$ are the edges, and $M:N \rightarrow V$ is a marking total function that assigns to all the nodes in $N$ a value in the domain $V = \{\mathit{Wrong},\mathit{Undefined}\}$.
\end{definition}

Initially, all nodes in the MET are marked as $\mathit{Undefined}$. But with every answer of the user, a new MET is produced. Concretely, given a MET $T = (N, E, M)$ and a node $n \in N$, the answer of the user to the question in $n$ produces a new MET such that: (i) if the answer is YES, then this node and its subtree is removed from the MET. (ii) If the answer is NO, then, all the nodes in the MET are removed except this node and its descendants.\footnote{It is also possible to accept \emph{I don't know} as an answer of the user. In this case, the debugger simply selects another node \cite{IS10}. For simplicity, we assume here that the user only answers $\mathit{Correct}$ or $\mathit{Wrong}$.} Therefore, note that the only node that can be marked as $\mathit{Wrong}$ is the root. Moreover, the rest of nodes can only be marked as $\mathit{Undefined}$ because when the answer is YES, the associated subtree is deleted from the MET.

Therefore, the size of the MET is gradually reduced with the answers. If we delete all nodes in the MET then the debugger concludes that no bug has been found. If, contrarily, we finish with a MET composed of a single node marked as wrong, this node is called the \emph{buggy node} and it is pointed to as being responsible for the bug of the program.

All this process is defined in Algorithm~\ref{algoritmoGeneral} where function \emph{selectNode} selects a node in the MET to be asked to the user with function \emph{askNode}. Therefore, \emph{selectNode} is the central point of this paper. In the rest of this section, we assume that \emph{selectNode} implements D\&Q. In the following we use $E^*$ to refer to the reflexive and transitive closure of $E$ and $E^+$ for the transitive closure.

\begin{algorithm}
\caption{General algorithm for algorithmic debugging}
\label{algoritmoGeneral}
\begin{algorithmic}
\STATE \textbf{Input:} A MET $T = (N, E, M)$\\
\STATE \textbf{Output:} A buggy node or $\bot$ if no buggy node is detected\\
\STATE \textbf{Preconditions:} $\forall n \in N$, $M(n) = \mathit{Undefined}$
\STATE \textbf{Initialization:} buggyNode $=\bot$

\medskip
\textbf{begin}

\medskip
\STATE (1) ~\textbf{do}
\STATE (2) ~~~~~node = selectNode($T$)\\
\STATE (3) ~~~~~answer = askNode(node)\\
\STATE (4) ~~~~~\textbf{if} (answer = $\mathit{Wrong}$) 
\STATE (5) ~~~~~\textbf{then} $M$(node) = $\mathit{Wrong}$
\STATE (6) ~~~~~~~~~~~~~buggyNode = node
\STATE (7) ~~~~~~~~~~~~~$N = \{n \in N \mid ($node $ \rightarrow n) \in E^*$\}
\STATE (8) ~~~~~\textbf{else} $N = N \backslash \{n \in N \mid ($node $ \rightarrow n) \in E^*$\}\\
\STATE (9) ~\textbf{while} $(\exists n \in N, M(n)=\mathit{Undefined})$
\STATE (10) \textbf{return} buggyNode\\
\medskip
\textbf{end}

\end{algorithmic}
\end{algorithm}

Both D\&Q by Shapiro and D\&Q by Hirunkitti assume that the individual weight of a node is always 1. Therefore, given a MET $T = (N, E, M)$, the weight of the subtree rooted at node $n \in N$, $w_n$, is defined recursively as its number of descendants including itself (i.e., $1 + \sum{\{w_{n'} \mid (n \rightarrow n') \in E\}})$.

D\&Q tries to simulate a dichotomic search by selecting the node that better divides the MET into two subMETs with a weight as similar as possible. Therefore, given a MET with $n$ nodes, D\&Q searches for the node whose weight is closer to $\frac{n}{2}$. 
The original algorithm by Shapiro always selects:
\begin{itemize}
\item the heaviest node $n'$ whose weight is as close as possible to $\frac{n}{2}$ with $w_{n'} \leq \frac{n}{2}$ 
\end{itemize}
Hirunkitti and Hogger noted that this is not enough to divide the MET by half and their improved version always selects the node whose weight is closer to $\frac{n}{2}$ between:
\begin{itemize}
\item the heaviest node $n'$ whose weight is as close as possible to $\frac{n}{2}$ with $w_{n'} \leq \frac{n}{2}$, or
\item the lightest node $n'$ whose weight is as close as possible to $\frac{n}{2}$ with $w_{n'} \geq \frac{n}{2}$
\end{itemize}


Because it is better, in the rest of the article we only consider Hirunkitti's D\&Q and refer to it as D\&Q.

\subsection{Limitations of D\&Q}

In this section we show that D\&Q is suboptimal when the MET does not contain a wrong node (i.e., all nodes are marked as undefined).\footnote{Modern debuggers \cite{IS10} allow the programmer to debug the MET while it is being generated. Thus the root node of the subtree being debugged is not necessarily marked as \emph{Wrong}.}
The intuition beyond this limitation is that the objective of D\&Q is to divide the tree by two, but the real objective should be to reduce the number of questions to be asked to the programmer.
For instance, consider the MET in Figure~\ref{samples} (left) where the black node is marked as wrong and D\&Q would select the gray node. The objective of D\&Q is to divide the 8 nodes into two groups of 4. Nevertheless, the real motivation of dividing the tree should be to divide the tree into two parts that would produce the same number of remaining questions (in this case 3). 

The problem comes from the fact that D\&Q does not take into account the marking of wrong nodes. For instance, observe the two METs in Figure~\ref{samples} (center) where each node is labeled with its weight and the black node is marked as wrong. In both cases D\&Q would behave exactly in the same way, because
it completely ignores the marking of the root. Nevertheless, it is evident that we do not need to ask again for a node that is already marked as wrong to determine whether it is buggy. 
However, D\&Q counts the nodes marked as wrong as part of their own weight, and this is a source of inefficiency.

\begin{figure}[h!]
\centering
\includegraphics[height=2.75cm]{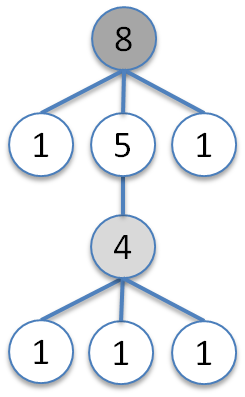}
\hspace{1cm}
\includegraphics[height=2cm]{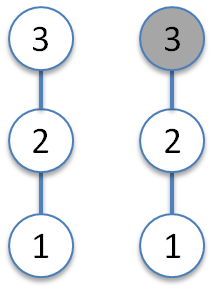}
\hspace{1cm}
\includegraphics[height=2.6cm]{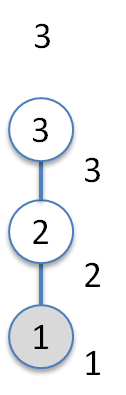}
\includegraphics[height=2.6cm]{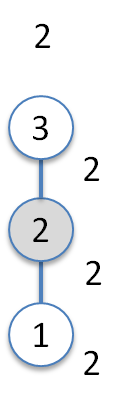}

\caption{Behavior of Divide and Query}
\label{samples}
\end{figure}

In the METs of Figure~\ref{samples} (center) we have two METs. In the one at the right nodes with weight 1 and 2 are optimal, but in the one at the left, only the node with weight 2 is optimal.
In both METs D\&Q would select either the node with weight 1 or the node with weight 2 (both are equally close to $\frac{3}{2}$). However, we show in Figure~\ref{samples} (right) that selecting node 1 is suboptimal, and the strategy should always select node 2. Considering that the gray node is the first node selected by the strategy, then the number at the side of a node represents the number of questions needed to find the bug if the buggy node is this node.  The number at the top of the figure represents the number of questions needed to determine that there is not a bug. Clearly, as an average, it is better to select first the node with weight 2 because we would perform less questions ($\frac{8}{4}$ vs. $\frac{9}{4}$ considering all four possible cases).

Therefore, D\&Q returns a set of nodes that contains the best node, but it is not able to determine which of them is the best node, thus being suboptimal when it is not selected. In addition, the METs in Figure \ref{5nodos} show that D\&Q is incomplete. Observe that the METs have 5 nodes, thus D\&Q would always select the node with weight 2. However, the node with weight 4 is equally optimal (both need $\frac{16}{6}$ questions as an average to find the bug) but it will be never selected by D\&Q because its weight is far from the half of the tree $\frac{5}{2}$.

\begin{figure}
\centering
\includegraphics[height=3.3cm]{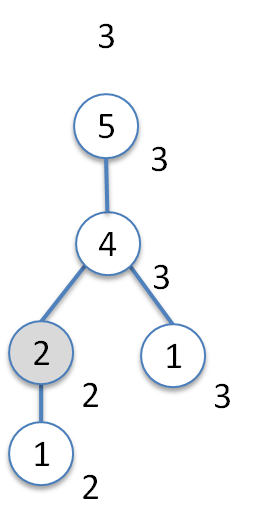}
\hspace{1cm}
\includegraphics[height=3.3cm]{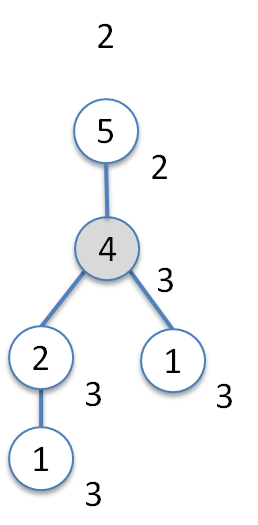}
\caption{Incompleteness of Divide and Query}
\label{5nodos}
\end{figure}

Another limitation of D\&Q is that it was designed to work with METs where all the nodes have the same individual weight, and moreover, this weight is assumed to be one. If we work with METs where nodes can have different individual weights and these weights can be any value greater or equal to zero, then D\&Q is suboptimal 
as it is demonstrated by the MET in Figure~\ref{fig_decimal}. In this MET, D\&Q would select node $n_1$ because its weight is closer to $\frac{21}{2}$ than any other node. However, node $n_2$ is the node that better divides the tree in two parts with the same probability of containing the bug.

\begin{figure}
\centering
\includegraphics[width=5.5cm]{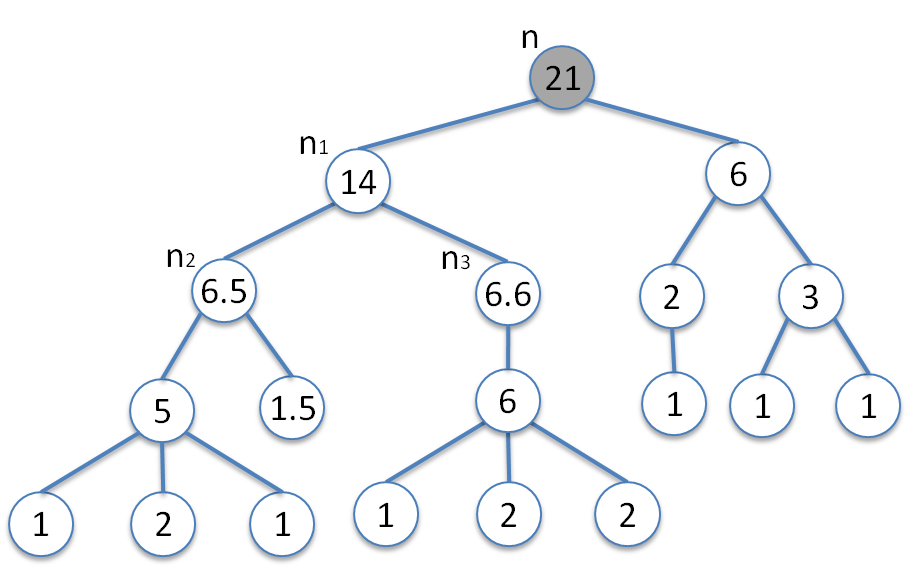}
\caption{MET with decimal individual weights}
\label{fig_decimal}
\end{figure}

In summary, 
(1) D\&Q is suboptimal when the MET is free of wrong nodes, 
(2) D\&Q is correct when the MET contains wrong nodes and all the nodes of the MET have the same weight, but
(3) D\&Q is suboptimal when the MET contains wrong nodes and the nodes of the MET have different individual weights.

\section{Optimal D\&Q}
\label{sec_optimal}

In this section we introduce a new version of D\&Q that tries to divide the MET into two parts with the same probability of containing the bug (instead of two parts with the same weight). We introduce new algorithms that are correct and complete even if the MET contains nodes with different individual weights. 
For this, we define the \emph{search area} of a MET as the set of undefined nodes.

\begin{definition} [Search area]
Let $T = (N, E, M)$ be a MET.  The \emph{search area} of $T$, $Sea(T)$, is defined as \{$n \in N \mid M(n) = \mathit{Undefined}$\}.
\end{definition}

While D\&Q uses the whole $T$, we only use $Sea(T)$, because answering all nodes in $Sea(T)$ guarantees that we can discover all buggy nodes \cite{Llo97}. 
Moreover, in the following we refer to the individual weight of a node $n$ with $wi_n$; and we refer to the weight of a (sub)tree rooted at $n$ with $w_n$ that is recursively defined as:
\[w_n=\left\{
\begin{array}{lll}
   \sum{\{w_{n'} \mid (n \rightarrow n') \in E\}} & ~~~~~~ & {\rm if }~ M(n) \neq \mathit{Undefined} \\      
   wi_n + \sum{\{w_{n'} \mid (n \rightarrow n') \in E\}} & ~~~~ & {\rm otherwise}
\end{array}\right.\]

Note that, contrarily to standard D\&Q, the definition of $w_n$ excludes those nodes that are not in the search area (i.e., the root node when it is wrong). Note also that $wi_n$ allows us to assign any individual weight to the nodes. This is an important generalization of D\&Q where it is assumed that all nodes have the same individual weight and it is always 1. 

\subsection{Debugging ETs where all nodes have the same individual weight \bf$wi \in \cR^+$}
\label{sec_sameWeight}

For the sake of clarity, given a node $n \in \mathit{Sea}(T)$, we distinguish between three subareas of $\mathit{Sea}(T)$ induced by $n$: (1) $n$ itself, whose individual weight is $wi_n$; (2) descendants of $n$, whose weight is

~~~~~~~~~$\mathit{Down}(n) = \sum{\{wi_{n'} \mid n' \in \mathit{Sea}(T) \land (n \rightarrow n') \in E^+\}}$

\noindent and (3) the rest of nodes, whose weight is

~~~~~~~~~~~~$\mathit{Up}(n) = \sum{\{wi_{n'} \mid n' \in \mathit{Sea}(T) \land (n \rightarrow n') \not\in E^*\}}$

\begin{example}
Consider the MET in Figure~\ref{gpf}.
\begin{figure}
\centering
\includegraphics[width=3cm]{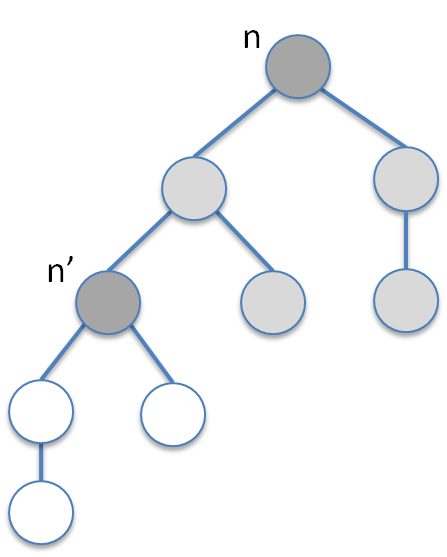}
\caption{Functions Up and Down}
\label{gpf}
\end{figure}
Assuming that the root $n$ is marked as wrong and all nodes have an individual weight of 1, then $\mathit{Sea}(T)$ contains all nodes except $n$, $\mathit{Up}(n')=4$ (total weight of the gray nodes), and $\mathit{Down}(n')=3$ (total weight of the white nodes).
\end{example}

Clearly, for any MET whose root is $n$ and a node $n'$, $M(n') = \mathit{Undefined}$, we have that:

\smallskip
$
\begin{array}{l@{~~~~~~~~~~~~~~~~~~~~~~~~~~~}r}
w_n = \mathit{Up}(n') + \mathit{Down}(n') + wi_{n'} & $(Equation 1)$\\
w_{n'} = \mathit{Down}(n') + wi_{n'} & $(Equation 2)$\\

\end{array}
$
\smallskip

Intuitively, given a node $n$, what we want to divide by half is the area formed by $\mathit{Up}(n)+\mathit{Down}(n)$. That is, $n$ will not be part of $\mathit{Sea}(T)$ after it has been answered, thus the objective is to make $\mathit{Up}(n)$ equal to $\mathit{Down}(n)$. This is another important difference with traditional D\&Q: $wi_n$ should not be considered when dividing the MET. We use the notation $n_1 \gg n_2$ to express that $n_1$ divides $\mathit{Sea}(T)$ better than $n_2$ (i.e., $|\mathit{Down}(n_1) - \mathit{Up}(n_1)| < |\mathit{Down}(n_2) - \mathit{Up}(n_2)|$). And we use $n_1 \equiv n_2$ to express that $n_1$ and $n_2$ equally divide $\mathit{Sea}(T)$. If we find a node $n$ such that $\mathit{Up}(n)=\mathit{Down}(n)$ then $n$ produces an optimal division, and should be selected by the strategy. If an optimal solution cannot be found, the following theorem states how to compare the nodes in order to decide which of them should be selected.



\begin{theorem}
\label{equation}
Given a MET $T = (N, E, M)$ whose root is $n \in N$, where $\forall n',n'' \in N, wi_{n'}=wi_{n''}$ and $\forall n' \in N, wi_{n'} > 0$, and given two nodes $n_1,n_2 \in Sea(T)$, with $w_{n_1} > w_{n_2}$, $n_1 \gg n_2$ if and only if $w_n > w_{n_1} + w_{n_2} - wi_n$.
\end{theorem}

\begin{theorem}
\label{prop}
Given a MET $T = (N, E, M)$ whose root is $n \in N$, where $\forall n',n'' \in N, wi_{n'}=wi_{n''}$ and $\forall n' \in N, wi_{n'} > 0$, and given two nodes $n_1,n_2 \in Sea(T)$, with $w_{n_1} > w_{n_2}$, $n_1 \equiv n_2$ if and only if $w_n = w_{n_1} + w_{n_2} - wi_n$.
\end{theorem}

Theorem \ref{equation} is useful when one node is heavier than the other. In the case that both nodes have the same weight, then the following theorem guarantees that they both equally divide the MET in all situations.

\begin{theorem}
\label{equation2}
Let $T = (N, E, M)$ be a MET where $\forall n,n' \in N, wi_n=wi_{n'}$ and $\forall n \in N, wi_n > 0$, and let $n_1, n_2 \in Sea(T)$ be two nodes, if $w_{n_1} = w_{n_2}$ then $n_1 \equiv n_2$.
\end{theorem}

\begin{corollary}
\label{Cor_prod}
Given a MET $T = (N, E, M)$ where $\forall n,n' \in N, wi_n=wi_{n'}$ and $\forall n \in N, wi_n > 0$, and given a node $n \in \mathit{Sea}(T)$, then $n$ optimally divides $\mathit{Sea}(T)$ if and only if $\mathit{Up}(n)=\mathit{Down}(n)$.
\end{corollary}

While Corollary~\ref{Cor_prod} states the objective of optimal D\&Q (finding a node $n$ such that $\mathit{Up}(n)=\mathit{Down}(n)$), Theorems \ref{equation} and \ref{equation2} provide a method to approximate this objective (finding a node $n$ such that $|\mathit{Down}(n) - \mathit{Up}(n)|$ is minimum in $\mathit{Sea}(T)$).

\subsubsection{An algorithm for Optimal D\&Q.}

Theorems~\ref{equation} and \ref{prop} provide equation $w_n \geq w_{n_1} + w_{n_2} - wi_n$ to compare two nodes $n_1, n_2$ by efficiently determining $n_1 \gg n_2$, $n_1 \equiv n_2$ or $n_1 \ll n_2$. However, with only this equation, we should compare all nodes to select the best of them (i.e., $n$ such that $\nexists n', n' \gg n$). Hence, in this section we provide an algorithm that allows us to find the best node in a MET with a minimum set of node comparisons. 

Given a MET, Algorithm~\ref{alg_camino} efficiently determines the best node to divide $\mathit{Sea}(T)$ by half (in the following the \emph{optimal node}). In order to find this node, the algorithm does not need to compare all nodes in the MET. It follows a path of nodes from the root to the optimal node which is closer to the root producing a minimum set of comparisons.


\begin{algorithm}
\caption{Optimal D\&Q ---SelectNode in Algorithm~\ref{algoritmoGeneral}---}
\label{alg_camino}
\begin{algorithmic}
\STATE \textbf{Input:} A MET $T = (N, E, M)$ whose root is $n \in N$,\\
\STATE ~~~~~~~~~~$\forall n',n'' \in N, wi_{n'}=wi_{n''}$ and $\forall n' \in N, wi_{n'} > 0$
\STATE \textbf{Output:} A node $n_{\mathit{Optimal}} \in N$\\
\STATE \textbf{Preconditions:} $\exists n' \in N$, $M(n') = \mathit{Undefined}$

\medskip

\textbf{begin}
\smallskip
\STATE (1) ~Candidate $= n$
\STATE (2) ~$\mathbf{do}$
\STATE (3) ~~~~~Best = Candidate
\STATE (4) ~~~~~Children = $\{m ~|~($Best $ \rightarrow m) \in E\}$
\STATE (5) ~~~~~$\mathbf{if}$ (Children = $\emptyset$) $\mathbf{then~return}$ Best
\STATE (6) ~~~~~Candidate = $n' \mid \forall n''$ with $n',n'' \in $ Children$,$ $w_{n'} \geq w_{n''}$
\STATE (7) ~$\mathbf{while}$ $(w_{\mathit{Candidate}} > \frac{w_n}{2})$
\STATE (8) ~$\mathbf{if}$ $(M($Best$) = \mathit{Wrong})$ $\mathbf{then}$ $\mathbf{return}$ Candidate
\STATE (9) ~$\mathbf{if}$ $(w_n \geq w_{\mathit{Best}} + w_{\mathit{Candidate}} - wi_n)$ $\mathbf{then ~return}$ Best
\STATE (10) ~~~~~~~~~~~~~~~~~~~~~~~~~~~~~~~~~~~~~~~~~~~~~~$\mathbf{else ~return}$ Candidate
\smallskip
\STATE \textbf{end}

\end{algorithmic}
\end{algorithm}

\begin{example} \label{ex_path}
Consider the MET in Figure~\ref{camino2} where $\forall n \in N, wi_n = 1$ and $M(n) = \mathit{Undefined}$.
\begin{figure}
\centering
\includegraphics[width=4.5cm]{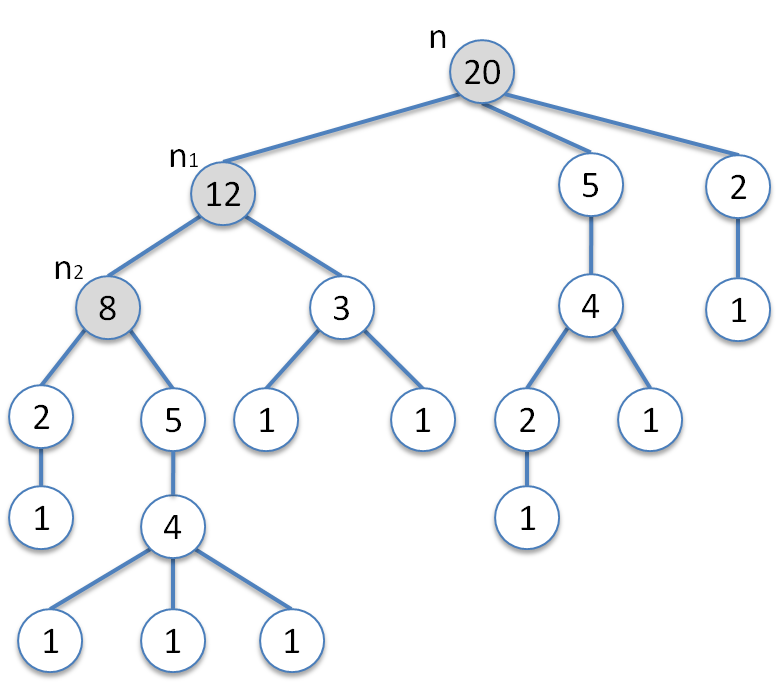}
\caption{Defining a path in a MET to find the optimal node}
\label{camino2}
\end{figure}
Observe that Algorithm~\ref{alg_camino} only needs to apply the equation in Theorem~\ref{equation} once to identify an optimal node. Firstly, it traverses the MET top-down from the root selecting at each level the heaviest node until we find a node whose weight is smaller than the half of the MET ($\frac{w_n}{2}$), thus, defining a path in the MET that is colored in gray. 
Then, the algorithm uses the equation $w_n \geq w_{n_1} + w_{n_2} - wi_n$  to compare nodes $n_1$ and $n_2$. Finally, the algorithm selects $n_1$. 
\end{example}

In order to prove the correctness of Algorithm~\ref{alg_camino}, we need to prove that (1) the node returned is really an optimal node, and (2) this node will always be found by the algorithm (i.e., it is always in the path defined by the algorithm).

The first point can be proven with Theorems~\ref{equation}, \ref{prop} and \ref{equation2}. 
The second point is the key idea of the algorithm and it relies on an interesting property of the path defined: while defining the path in the MET, only four cases are possible, and all of them coincide in that the subtree of the heaviest node will contain an optimal node.

In particular, when we use Algorithm~\ref{alg_camino} and compare two nodes $n_1,n_2$ in a MET whose root is $n$, we find four possible cases: \\

\noindent$~~~~~${\bf Case 1:} $n_1$ and $n_2$ are brothers.\\
$~~~~~~${\bf Case 2:} $w_{n_1} > w_{n_2} ~\land~ w_{n_2} > \frac{w_n}{2}$.\\
$~~~~~~${\bf Case 3:} $w_{n_1} > \frac{w_n}{2} ~\land~ w_{n_2} \leq \frac{w_n}{2}$.\\
$~~~~~~${\bf Case 4:} $w_{n_1} > w_{n_2} ~\land~ w_{n_1} \leq \frac{w_n}{2}$.
\bigskip
%

\begin{figure}
\centering
\includegraphics[width=2cm]{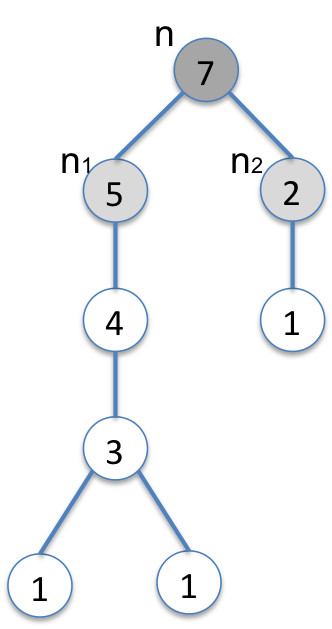}~
\includegraphics[width=2cm]{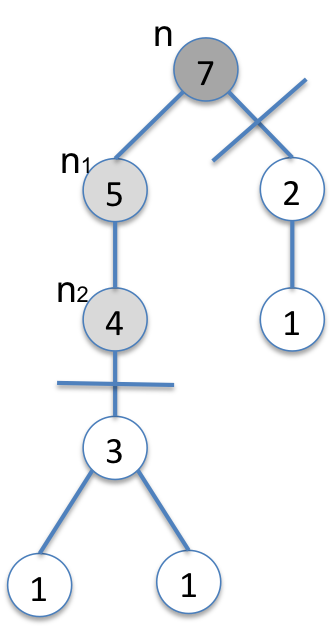}~
\includegraphics[width=2cm]{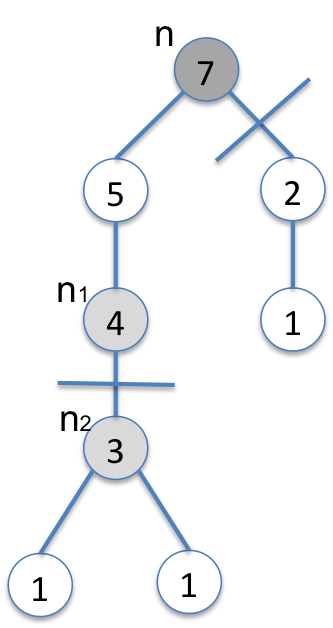}~
\includegraphics[width=2.2cm]{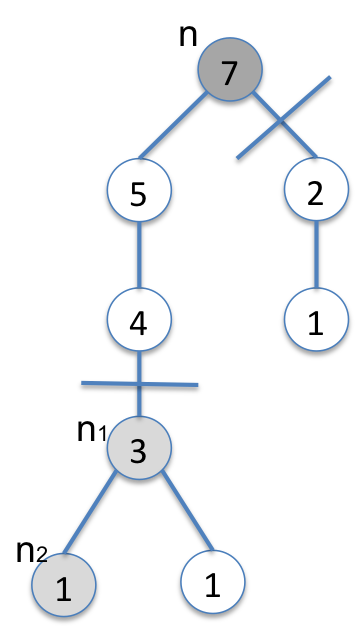}\\
~~~Case 1 ~~~~~~~~ Case 2 ~~~~~~~~ Case 3 ~~~~~~~~ Case 4
\caption{Determining the best node in a MET (four possible cases)}
\label{camino}
\end{figure}

We have proven---the individual proofs are part of the proof of Theorem~\ref{theo_corr}---that in cases 1 and 4, the heaviest node is better (i.e., if $w_{n_1}>w_{n_2}$ then $n_1 \gg n_2$); In case 2, the lightest node is better; and in case 3, the best node must be determined with the equations of Theorems~\ref{equation}, \ref{prop} and \ref{equation2}. Observe that these results allow the algorithm to determine the path to the optimal node that is closer to the root. 
For instance, in Example~\ref{ex_path} case 1 is used to select a child, e.g., node 12 instead of node 5 or node 2, and node 8 instead of node 3.  
Case 2 is used to go down and select node 12 instead of node 20. Case 4 is used to stop going down at node 8 because it is better than all its descendants. And it is also used to determine that nodes 2, 3 and 5 are better than all their descendants. Finally, case 3 is used to select the optimal node, 12 instead of 8. 
Note that D\&Q could have selected node 8 that is equally close to $\frac{20}{2}$ than node 12; but it is suboptimal because $\mathit{Up}(8)=12$ and $\mathit{Down}(8)=7$ whereas $\mathit{Up}(12)=8$ and $\mathit{Down}(12)=11$.

The correctness of Algorithm~\ref{alg_camino} is stated by the following theorem.

\begin{theorem}[Correctness]
\label{theo_corr}
Let $T = (N, E, M)$ be a MET where $\forall n,n' \in N, wi_n=wi_{n'}$ and $\forall n \in N, wi_n > 0$, then the execution of Algorithm~\ref{alg_camino} with $T$ as input always terminates producing as output a node $n \in Sea(T)$ such that $\nexists n' \in Sea(T) \mid n' \gg n$.
\end{theorem}

Algorithm~\ref{alg_camino} always returns a single optimal node. However, the equation in Theorem~\ref{equation} in combination with the equation in Theorem~\ref{prop} can be used to identify all optimal nodes in the MET.
This is implemented in Algorithm~\ref{alg_camino2} that is complete, and thus it returns nodes 2 and 4 in the MET of Figure~\ref{5nodos} where D\&Q can only detect node 2 as optimal.

%
%
%

\begin{algorithm}
\caption{Optimal D\&Q (Complete)  ---SelectNode in Algorithm~\ref{algoritmoGeneral}---}
\label{alg_camino2}
\begin{algorithmic}
\STATE \textbf{Input:} A MET $T = (N, E, M)$ whose root is $n \in N$,\\
\STATE ~~~~~~~~~~$\forall n',n'' \in N, wi_{n'}=wi_{n''}$ and $\forall n' \in N, wi_{n'} > 0$
\STATE \textbf{Output:} A set of nodes $O \subseteq N$\\
\STATE \textbf{Preconditions:} $\exists n' \in N$, $M(n') = \mathit{Undefined}$

\medskip

\textbf{begin}
\smallskip
\STATE (1) ~Candidate $= n$
\STATE (2) ~$\mathbf{do}$
\STATE (3) ~~~~~Best = Candidate
\STATE (4) ~~~~~Children = $\{m ~|~($Best $ \rightarrow m) \in E\}$
\STATE (5) ~~~~~$\mathbf{if}$ (Children = $\emptyset$) $\mathbf{then~return}$ $\{$Best$\}$
\STATE (6) ~~~~~Candidate = $n' \mid \forall n''$ with $n',n'' \in $ Children$,$ $w_{n'} \geq w_{n''}$
\STATE (7) ~$\mathbf{while}$ $(w_{\mathit{Candidate}} > \frac{w_n}{2})$
\STATE (8) ~Candidates = $\{n' \mid \forall n''$ with $n',n'' \in $ Children$,$ $w_{n'} \geq w_{n''}\}$
\STATE (9) ~$\mathbf{if}$ $(M($Best$) = \mathit{Wrong})$ $\mathbf{then}$ $\mathbf{return}$ Candidates
\STATE (10) ~$\mathbf{if}$ $(w_n > w_{\mathit{Best}} + w_{\mathit{Candidate}} - wi_n)$ $\mathbf{then ~return}$ $\{$Best$\}$\\
\STATE (11) ~$\mathbf{if}$ $(w_n = w_{\mathit{Best}} + w_{\mathit{Candidate}} - wi_n)$ $\mathbf{then ~return}$ $\{$Best$\} ~\cup$ Candidates\\
\STATE (12) ~~~~~~~~~~~~~~~~~~~~~~~~~~~~~~~~~~~~~~~~~~~~~~~$\mathbf{else ~return}$ Candidates
\smallskip
\STATE \textbf{end}

\end{algorithmic}
\end{algorithm}

\subsection{Debugging METs where nodes can have different individual weights in $\cR^+ \cup \{0\}$}
\label{sec_variableMET}


In this section we generalize Divide and Query to the case where nodes can have different individual weights and these weights can be any value greater or equal to zero. As shown in Figure~\ref{fig_decimal}, in this general case traditional D\&Q fails to identify the optimal node (it selects node $n_1$ but the optimal node is $n_2$). The algorithm presented in the previous section is also suboptimal when the individual weights can be different. For instance, in the MET of Figure~\ref{fig_decimal}, it would select node $n_3$. For this reason, in this section we introduce Algorithm~\ref{alg_general}, a general algorithm able to identify an optimal node in all cases. 
It does not mean that Algorithm~\ref{alg_camino} is useless. Algorithm~\ref{alg_camino} is optimal when all nodes have the same weight, and in that case, it is more efficient than Algorithm~\ref{alg_general}.
Theorem~\ref{theo_corr2} ensures the finiteness and correctness of Algorithm~\ref{alg_general}.

\begin{algorithm}
\caption{Optimal D\&Q General ---SelectNode in Algorithm~\ref{algoritmoGeneral}---}
\label{alg_general}
\begin{algorithmic}
\STATE \textbf{Input:} A MET $T = (N, E, M)$ whose root is $n \in N$ and $\forall n' \in N, wi_{n'} \geq 0$
\STATE \textbf{Output:} A node $n_{\mathit{Optimal}} \in N$\\
\STATE \textbf{Preconditions:} $\exists n' \in N$, $M(n') = \mathit{Undefined}$

\medskip

\textbf{begin}
\smallskip
\STATE (1) ~Candidate = $n$
\STATE (2) ~$\mathbf{do}$
\STATE (3) ~~~~~Best = Candidate
\STATE (4) ~~~~~Children = $\{m ~|~($Best $ \rightarrow m) \in E\}$
\STATE (5) ~~~~~$\mathbf{if}$ (Children = $\emptyset$) $\mathbf{then~return}$ Best
\STATE (6) ~~~~~Candidate = $n' \mid \forall n''$ with $n',n'' \in $ Children, $w_{n'} \geq w_{n''}$
\STATE (7) ~$\mathbf{while}$ $(w_{\mathit{Candidate}} - \frac{wi_{\mathit{Candidate}}}{2} > \frac{w_{n}}{2})$
\STATE (8) ~Candidate = $n' \mid \forall n''$ with $n',n'' \in $ Children$, w_{n'} - \frac{wi_{n'}}{2} \geq w_{n''} - \frac{wi_{n''}}{2}$
\STATE (9) ~$\mathbf{if}$ $(M($Best$) = \mathit{Wrong})$ $\mathbf{then}$ $\mathbf{return}$ Candidate
\STATE (10) ~$\mathbf{if}$ $(w_n \geq w_{\mathit{Best}} + w_{\mathit{Candidate}} - \frac{wi_{\mathit{Best}}}{2} - \frac{wi_{\mathit{Candidate}}}{2})$ $\mathbf{then ~return}$ Best\\
\STATE (11) ~~~~~~~~~~~~~~~~~~~~~~~~~~~~~~~~~~~~~~~~~~~~~~~~~~~~~~~~~~~~~~~~~~$\mathbf{else ~return}$ Candidate\\
\smallskip
\STATE \textbf{end}

\end{algorithmic}
\end{algorithm}

\begin{theorem}[Correctness]
\label{theo_corr2}
Let $T = (N, E, M)$ be a MET where $\forall n \in N, wi_n\geq 0$, then the execution of Algorithm~\ref{alg_general} with $T$ as input always terminates producing as output a node $n \in Sea(T)$ such that $\nexists n' \in Sea(T) \mid n' \gg n$.
\end{theorem}

\subsection{Debugging METs where nodes can have different individual weights in $\cR^+$}
\label{sec_variableMET2}

In the previous section we provided an algorithm that optimally selects an optimal node of the MET with a minimum set of node comparisons. But this algorithm is not complete due to the fact that we allow the nodes to have an individual weight of zero. For instance, when all nodes have an individual weight of zero, Algorithm~\ref{alg_general} returns a single optimal node, but it is not able to find all optimal nodes.

Given a node (say $n$), the difference between having an individual weight of zero, $wi_n$, and having a (total) weight of zero, $w_n$, should be clear. The former means that this node did not cause the bug, the later means that none of the descendants of this node (neither the node itself) caused the bug. Surprisingly, the use of nodes with individual weights of zero has not been exploited in the literature. Assigning a (total) weight of zero to a node has been used for instance in the technique called \emph{Trusting} \cite{LC07}. This technique allows the user to trust a method. When this happens all the nodes related to this method and their descendants are pruned from the tree (i.e., these nodes have a (total) weight of zero).

If we add the restriction that nodes cannot be assigned with an individual weight of zero, then we can refine Algorithm~\ref{alg_general} to ensure completeness. This refined version is Algorithm~\ref{alg_general2}.

\begin{algorithm}[h!]
\caption{Optimal D\&Q General (Complete) ---SelectNode in Algorithm~\ref{algoritmoGeneral}---}
\label{alg_general2}
\begin{algorithmic}
\STATE \textbf{Input:} A MET $T = (N, E, M)$ whose root is $n \in N$ and $\forall n' \in N, wi_{n'} > 0$
\STATE \textbf{Output:} A set of nodes $O \subseteq N$\\
\STATE \textbf{Preconditions:} $\exists n' \in N$, $M(n') = \mathit{Undefined}$

\medskip

\textbf{begin}
\smallskip
\STATE (1) ~Candidate = $n$
\STATE (2) ~$\mathbf{do}$
\STATE (3) ~~~~~Best = Candidate
\STATE (4) ~~~~~Children = $\{m ~|~($Best $ \rightarrow m) \in E\}$
\STATE (5) ~~~~~$\mathbf{if}$ (Children = $\emptyset$) $\mathbf{then~return}$ $\{$Best$\}$
\STATE (6) ~~~~~Candidate = $n' \mid \forall n''$ with $n',n'' \in $ Children, $w_{n'} \geq w_{n''}$
\STATE (7) ~$\mathbf{while}$ $(w_{\mathit{Candidate}} - \frac{wi_{\mathit{Candidate}}}{2} > \frac{w_{n}}{2})$
\STATE (8) ~Candidates = $\{n' \mid \forall n''$ with $n',n'' \in $ Children$, w_{n'} - \frac{wi_{n'}}{2} \geq w_{n''} - \frac{wi_{n''}}{2}\}$
\STATE (9) ~Candidate = $n' \in $ Candidates
\STATE (10) ~$\mathbf{if}$ $(M($Best$) = \mathit{Wrong})$ $\mathbf{then}$ $\mathbf{return}$ Candidates
\STATE (11) ~$\mathbf{if}$ $(w_n > w_{\mathit{Best}} + w_{\mathit{Candidate}} - \frac{wi_{\mathit{Best}}}{2} - \frac{wi_{\mathit{Candidate}}}{2})$ $\mathbf{then ~return}$ $\{$Best$\}$\\
\STATE (12) ~$\mathbf{if}$ $(w_n = w_{\mathit{Best}} + w_{\mathit{Candidate}} - \frac{wi_{\mathit{Best}}}{2} - \frac{wi_{\mathit{Candidate}}}{2})$ $\mathbf{then}$\\
\STATE ~~~~~~~~~~~~~~~~~~~~~~~~~~~~~~~~~~~~~~~~~~~~~~~~~~~~~~~~~~~~~~~~~~$\mathbf{return}$ $\{$Best$\} ~\cup$ Candidates\\
\STATE (13) ~~~~~~~~~~~~~~~~~~~~~~~~~~~~~~~~~~~~~~~~~~~~~~~~~~~~~~~~~~~~~~~~~~$\mathbf{else ~return}$ Candidates\\
\smallskip
\STATE \textbf{end}

\end{algorithmic}
\end{algorithm}


\section{Conclusion}
\label{sec_conclusions}

During three decades, Divide \& Query has been the more efficient algorithmic debugging strategy. On the practical side, all current algorithmic debuggers implement D\&Q \cite{Bra07,Cab06,Tho06,IS10,MCC07,Mac05,Nai89,nilphd98,BerniePopePhd}, and experiments \cite{Cab05,Sil09} (see also http://users.dsic.upv.es/$^\sim$jsilva/DDJ/\#Experiments) demonstrate that it performs on average 2-36\% less questions than other strategies. On the theoretical side, because D\&Q intends a dichotomic search, it has been thought optimal with respect to the number of questions performed, and thus research on algorithmic debugging strategies has focused on other aspects such as reducing the complexity of questions.

In this work we show that in some situations current algorithms for D\&Q are incomplete and inefficient because they are not able to find all optimal nodes, and sometimes they return nodes that are not optimal. We have identified the sources of inefficiency and provided examples that show both the incompleteness and incorrectness of the technique. 

The main contribution of this work is a new algorithm for D\&Q that is optimal in all cases; including a generalization of the technique where all nodes of the ET can have different individual weights in $\cR^+ \cup \{0\}$. The algorithm has been proved terminating and correct. And a slightly modified version of the algorithm has been provided that returns all optimal solutions, thus being complete.

We have implemented the technique and experiments show that it is more efficient than all previous algorithms (see column {\tt D\&QO} in Figure~\ref{fig_comparison}). The im\-ple\-men\-ta\-tion---including the source code---and the experiments are publicly available at: http://users.dsic.upv.es/$^\backsim$jsilva/DDJ.

\bibliographystyle{abbrv}
\bibliography{biblio}

\newpage

\appendix\label{appendix}

\section{Proofs of Technical Results\label{sec-proofs}}




In this section, for the sake of clarity, we use $u_n$ and $d_n$ instead of $\mathit{Up}(n)$ and $\mathit{Down}(n)$ respectively.
Moreover, we distinguish between two kinds of METs to prove the theorems of sections~\ref{sec_sameWeight} and \ref{sec_variableMET} respectively.

\begin{definition} [Uniform MET]
A \emph{uniform MET} $T = (N, E, M)$ is a MET, where $\forall n,n' \in N, wi_n=wi_{n'}$ and $\forall n \in N, wi_n > 0$.
\end{definition}

\begin{definition} [Variable MET]
A \emph{variable MET} $T = (N, E, M)$ is a MET, where $\forall n \in N, wi_n \geq 0$.
\end{definition}

\subsection{Proofs of Theorems~\ref{equation}, \ref{prop} and~\ref{equation2}}

Here, we prove Theorems~\ref{equation}, \ref{prop} and~\ref{equation2} that are used in Algorithm~\ref{alg_camino} to compare nodes of the MET and determine which of them is better.
For the proof of Theorem~\ref{equation}, we need to prove first the following lemma.

\begin{lemma} \label{lemma_UD}
Let $T = (N, E, M)$ be a uniform MET whose root is $n \in N$, and let $n_1,n_2 \in Sea(T)$ be two nodes. Then, $n_1 \gg n_2$ if and only if $u_{n_1}*d_{n_1} > u_{n_2}*d_{n_2}$.
\end{lemma}

\begin{proof}
We prove that $u_{n_1}*d_{n_1} > u_{n_2}*d_{n_2}$ implies that $|d_{n_1} - u_{n_1}| < |d_{n_2} - u_{n_2}|$ and vice versa. This can be shown by developing the equation $u_{n_1}*d_{n_1} > u_{n_2}*d_{n_2}$.\\
Firstly, note that $w_n = \sum{\{wi_{n'} \mid n' \in Sea(T)\}}$, then by Equation 1 we know that $w_n = u_{n_1}+d_{n_1}+wi_{n_1} = u_{n_2}+d_{n_2}+wi_{n_2}$.
Therefore, as $wi_{n_1} = wi_{n_2} = wi_{n}$ the optimal division of $Sea(T)$ happens when $u_{n_1}=d_{n_1}=\frac{w_{n}-wi_{n}}{2}$.
For the sake of simplicity in the notation, let  $c=\frac{w_{n}-wi_{n}}{2}$ and let $h_1=c-d_{n_1}=u_{n_1}-c$ and $h_2=c-d_{n_2}=u_{n_2}-c$. Then,


{\footnotesize
$
\begin{array}{l@{}r}
u_{n_1}*d_{n_1} > u_{n_2}*d_{n_2}\\
$Therefore, we replace $u_{n_1}$, $d_{n_1}$, $u_{n_2}$ and $d_{n_2}$:$\\
(c + h_1) * (c - h_1) > (c + h_2) * (c - h_2)\\
c^2 - h_1 * c + h_1 * c - h_1^2 > c^2 - h_2 * c + h_2 * c - h_2^2\\
$We simplify:$\\
c^2 - h_1^2 > c^2 - h_2^2\\
-h_1^2 > -h_2^2\\
h_1^2 < h_2^2\\ 
$And finally we obtain that:$\\
|h_1| < |h_2|
\end{array}
$
}

\noindent Hence, if the product $u_{n_1}*d_{n_1}$ is greater than $u_{n_2}*d_{n_2}$ then $|h_1|<|h_2|$ and thus, because  $h_1$ and $h_2$ represent distances to the center, $n_1 \gg n_2$.
\end{proof}

%
%

%
%

\noindent {\bf Theorem~\ref{equation}.} 
\emph{
Given a uniform MET $T = (N, E, M)$ whose root is $n \in N$, and given two nodes $n_1, n_2 \in Sea(T)$, with $w_{n_1} > w_{n_2}$,
$n_1 \gg n_2$ if and only if $w_n > w_{n_1} + w_{n_2} - wi_n$.
}

\begin{proof}
By Lemma~\ref{lemma_UD} we know that if $u_{n_1} * d_{n_1} > u_{n_2} * d_{n_2}$ then $n_1 \gg n_2$. Thus it is enough to prove that $w_n > w_{n_1} + w_{n_2} - wi_n$ implies $u_{n_1} * d_{n_1} > u_{n_2} * d_{n_2}$
and vice versa when $w_{n_1} > w_{n_2}$.\\

{\footnotesize
$
\begin{array}{l}
w_n > w_{n_1} + w_{n_2} - wi_n\\
$Adding $wi_n - wi_n$:$\\
w_n > w_{n_1} + w_{n_2} - 2*wi_n + wi_n\\ 
$We replace $w_{n_1}$, $w_{n_2}$ by Equation 2:$\\
w_n > d_{n_1} + d_{n_2} + wi_n\\
$Adding $wi_n * d - wi_n * d$:$\\
w_n > d_{n_1} + d_{n_2} + wi_n * d + wi_n - wi_n * d\\
w_n > d_{n_1} + d_{n_2} + wi_n * d + wi_n (1 - d)\\
$Using $d = \frac{d_{n_1}}{d_{n_1} - d_{n_2}}$ we get:$\\
w_n > d_{n_1} + d_{n_2} + wi_n\frac{d_{n_1}}{d_{n_1} - d_{n_2}} + wi_n (1 - \frac{d_{n_1}}{d_{n_1} - d_{n_2}})\\
w_n > d_{n_1} + d_{n_2} + wi_n\frac{d_{n_1}}{d_{n_1} - d_{n_2}} + wi_n(\frac{d_{n_1} - d_{n_2}}{d_{n_1} - d_{n_2}} - \frac{d_{n_1}}{d_{n_1} - d_{n_2}})\\
w_n > d_{n_1} + d_{n_2} + wi_n\frac{d_{n_1}}{d_{n_1} - d_{n_2}} + wi_n\frac{-d_{n_2}}{d_{n_1} - d_{n_2}}\\
w_n > d_{n_1} + d_{n_2} + wi_n\frac{d_{n_1}}{d_{n_1} - d_{n_2}} - wi_n\frac{d_{n_2}}{d_{n_1} - d_{n_2}}\\
$Because $d_{n_1} + d_{n_2} = \frac{d_{n_1}^2 - d_{n_2}^2}{d_{n_1} - d_{n_2}}$ then:$\\ 
w_n > \frac{d_{n_1}^2 - d_{n_2}^2}{d_{n_1} - d_{n_2}} + \frac{d_{n_1} * wi_n}{d_{n_1} - d_{n_2}} - \frac{d_{n_2} * wi_n}{d_{n_1} - d_{n_2}}\\ 
$Because $w_{n_1} > w_{n_2}$ we know by Equation 2 that $ d_{n_1} - d_{n_2} > 0$, thus:$\\
(d_{n_1} - d_{n_2}) * w_n > d_{n_1}^2 - d_{n_2}^2 + d_{n_1} * wi_n - d_{n_2} * wi_n\\ 
d_{n_1} * w_n - d_{n_2} * w_n > d_{n_1}^2 - d_{n_2}^2 + d_{n_1} * wi_n - d_{n_2} * wi_n\\
d_{n_1} * w_n - d_{n_1}^2 - d_{n_1} * wi_n > d_{n_2} * w_n - d_{n_2}^2 - d_{n_2} * wi_n\\
d_{n_1} * (w_n - d_{n_1} - wi_n) > d_{n_2} * (w_n - d_{n_2} - wi_n)\\
$As $wi_{n} = wi_{n_1} = wi_{n_2}$ we replace $w_n - d_{n_1} - wi_n$, $w_n - d_{n_2} - wi_n$ by Equation 1:$\\
d_{n_1} * u_{n_1} > d_{n_2} * u_{n_2}\\
\end{array}
$\\
}
\end{proof}

\noindent {\bf Theorem~\ref{prop}.}
\emph{
Given a uniform MET $T = (N, E, M)$ whose root is $n \in N$, and given two nodes $n_1, n_2 \in Sea(T)$, with $w_{n_1} > w_{n_2}$, $n_1 \equiv n_2$ if and only if $w_n = w_{n_1} + w_{n_2} - wi_n$.
}

\begin{proof}
The proof is completely analogous to the proof of Theorem \ref{equation}. The only difference is that the equation that is developed should be $w_n = w_{n_1} + w_{n_2} - wi_n$.
\end{proof}

\noindent {\bf Theorem~\ref{equation2}.}
\emph{
Let $T = (N, E, M)$ be a uniform MET, and let $n_1, n_2 \in Sea(T)$ be two nodes, if $w_{n_1} = w_{n_2}$ then $n_1 \equiv n_2$.
}

\begin{proof}
We prove that $w_{n_1} = w_{n_2}$ implies $|d_{n_1} - u_{n_1}| = |d_{n_2} - u_{n_2}|$ and thus $n_1 \equiv n_2$:\\

{\footnotesize
$
\begin{array}{l@{\hspace{-2cm}}r}
w_{n_1} = w_{n_2} & $we replace $w_{n_1}, w_{n_2}$ by Equation 2$\\
d_{n_1} + wi_{n_1} = d_{n_2} + wi_{n_2} & $using $wi_{n_1} = wi_{n_2}\\
d_{n_1} = d_{n_2} & $using $w_{n_1} = w_{n_2}\\
w_{n_1} - w_n + d_{n_1} = w_{n_2} - w_n + d_{n_2} & $replacing $w_{n_1}, w_{n_2}$ by Equation 2$\\
(d_{n_1} + wi_{n_1}) - (u_{n_1} + d_{n_1} + wi_{n_1}) + d_{n_1} & $and $w_{n}$ by Equation 1$\\ 
~~~~~~~~~~~~~~~~~~~~~~= (d_{n_2} + wi_{n_2}) - (u_{n_2} + d_{n_2} + wi_{n_2}) + d_{n_2} & $we simplify$\\
d_{n_1} - u_{n_1} = d_{n_2} - u_{n_2}\\
|d_{n_1} - u_{n_1}| = |d_{n_2} - u_{n_2}|\\
\end{array}
$\\
}
\end{proof}

\noindent {\bf Corollary~\ref{Cor_prod}.}
\emph{
Given a uniform MET $T = (N, E, M)$, and given a node $n \in Sea(T)$, then $n$ optimally divides $Sea(T)$ if and only if $u_{n}=d_{n}$.
}

\begin{proof}
If $n$ optimally divides $Sea(T)$ then the product $u_{n}*d_{n}$ is maximum, and there does not exist other node $n' \in Sea(T)$ such that $u_{n'}*d_{n'} > u_{n}*d_{n}$. This can be easily shown taking into account that the figure of the product is a parabola whose vertex is the maximum value. Therefore, we can compute the maximum by deriving the product.

For simplicity, let $prod=u_{n}*d_{n}$ and $sum=u_{n}+d_{n}$. Then, we start by transforming the equation $u_{n}*d_{n}$ in such a way that it only depends on one of the factors (e.g., $u_{n}$):

{\footnotesize
$
\begin{array}{lr}
u_{n}*d_{n}=prod\\
$We replace $d_{n}:\\
u_{n} * (sum - u_{n}) = prod\\
u_{n} * sum  - u_{n}^2 = prod \\ 
$We derive the equation and equate it to zero:$\\
\frac{d}{du_{n}}(u_{n} * sum  - u_{n}^2) = 0\\
sum - 2 u_{n} = 0\\
$And finally we get the value of $u_{n}$ in the vertex:$\\ 
u_{n} = \frac{sum}{2}\\
\end{array}
$
}

\noindent Now, we can infer $d_{n}$ from $u_{n}$ by simply replacing the value of $u_{n}$ in the equation $u_{n} + d_{n} = sum$:

$
\begin{array}{l@{~~~~~~}r}
\frac{sum}{2} + d_{n} = sum\\
d_{n} = sum - \frac{sum}{2}\\
d_{n} = \frac{sum}{2}\\
d_{n} = u_{n}
\end{array}
$\\
\end{proof}

\newpage
\subsection{Proof of Theorem~\ref{theo_corr}}


Theorem~\ref{theo_corr} states the correctness of Algorithm~\ref{alg_camino} used when all nodes have the same individual weight.
Firstly, we proof the following auxiliary lemma.

\begin{lemma}
\label{comparacion}
Let $T = (N, E, M)$ be a uniform MET whose root is $n \in N$ and $n_1, n_2 \in Sea(T)$ with $w_{n_1} > w_{n_2}$, if $w_n \geq w_{n_1} + w_{n_2}$ then $n_1 \gg n_2$.
\end{lemma}

\begin{proof}
Firstly, by Theorem~\ref{equation} we know that if $w_n > w_{n_1} + w_{n_2} - wi_n$ when $w_{n_1} > w_{n_2}$ then $n_1 \gg n_2$.
Therefore, as $wi_n > 0$, if $w_n \geq w_{n_1} + w_{n_2}$ then $w_n > w_{n_1} + w_{n_2} - wi_n$ and hence $n_1 \gg n_2$.
\end{proof}

In order to prove the correctness of Algorithm~\ref{alg_camino}, we also need to prove the four cases presented in Section~\ref{sec_sameWeight} that are used in the algorithm: \\
$~~~~~~${\bf Case 1:} $n_1$ and $n_2$ are brothers.\\
$~~~~~~${\bf Case 2:} $w_{n_1} > w_{n_2} ~\land~ w_{n_2} > \frac{w_n}{2}$.\\
$~~~~~~${\bf Case 3:} $w_{n_1} > \frac{w_n}{2} ~\land~ w_{n_2} \leq \frac{w_n}{2}$.\\
$~~~~~~${\bf Case 4:} $w_{n_1} > w_{n_2} ~\land~ w_{n_1} \leq \frac{w_n}{2}$.

We prove each case in a separate lemma. 
In case 1, the following lemma shows that given two brother nodes $n_1$ and $n_2$, then the heaviest node is better. 

\begin{lemma}
\label{lem_caminoa}
Given a uniform MET $T = (N, E, M)$ whose root is $n \in N$ and given three nodes $n_1 \in N$ and $n_2, n_3 \in Sea(T)$ with $(n \rightarrow n_1) \in E^*$,$(n_1 \rightarrow n_2),(n_1 \rightarrow n_3) \in E$, $n_2 \gg n_3 \lor n_2 \equiv n_3$ if and only if $w_{n_2} \geq w_{n_3}$.
\end{lemma}

\begin{proof}
We prove first that $w_{n_2} \geq w_{n_3}$ implies $n_2 \gg n_3 \lor n_2 \equiv n_3$:
Trivially, $w_n \ge w_{n_2} + w_{n_3}$ because $n_2$ and $n_3$ are children of $n_1$ and $n_1$ is descendant of $n$. Therefore, by Lemma~\ref{comparacion} and Theorem~\ref{equation2}, $n_2 \gg n_3 \lor n_2 \equiv n_3$.
Now, we prove that $n_2 \gg n_3 \lor n_2 \equiv n_3$ implies $w_{n_2} \geq w_{n_3}$:
We prove it by contradiction assuming that $w_{n_2} < w_{n_3}$ when $n_2 \gg n_3 \lor n_2 \equiv n_3$, and proving that when $w_{n_2} < w_{n_3}$ and $n_2 \gg n_3 \lor n_2 \equiv n_3$, neither $w_n > w_{n_2} + w_{n_3} - wi_n$ nor $w_n \leq w_{n_2} + w_{n_3} - wi_n$ holds.
By Theorem~\ref{equation} $w_n > w_{n_2} + w_{n_3} - wi_n$ is false because $n_2 \gg n_3 \lor n_2 \equiv n_3$. Moreover, because $n_2$ and $n_3$ are brothers, we know that $w_n \geq w_{n_2} + w_{n_3}$, and hence  $w_n \leq w_{n_2} + w_{n_3} - wi_n$ is also false.
\end{proof}

In case 2, the following lemma ensures that given two nodes $n_1$ and $n_2$ such that $n_1\rightarrow n_2$, if $w_{n_2} > \frac{w_n}{2}$ then $n_2$ is better.

\begin{lemma}
\label{lem_caminoc}
Given a uniform MET $T = (N, E, M)$ whose root is $n \in N$, and given two nodes $n_1, n_2 \in Sea(T)$, with $(n_1\rightarrow n_2) \in E$,  if $w_{n_2} > \frac{w_n}{2}$ then $n_2 \gg n_1$.
\end{lemma}

\begin{proof}
We prove the lemma by contradiction assuming that $n_1 \gg n_2$ or $n_1 \equiv n_2$.
First, we know that $w_{n_2} = \frac{w_n}{2} + inc_{n_2}$ with $inc_{n_2} > 0$. And we know that $w_{n_1} = \frac{w_n}{2} + inc_{n_2} + wi_n + inc_{n_1}$ with $inc_{n_1} \geq 0$, where $inc_{n_1}$ represent the weight of the possible brothers of $n_2$.
By Theorems~\ref{equation} and \ref{prop} we know that $w_n \geq w_{n_1} + w_{n_2} - wi_{n}$ when $w_{n_1} > w_{n_2}$ implies $n_1 \gg n_2 \lor n_1 \equiv n_2$.

{\footnotesize
$
\begin{array}{l@{~~~~~~}r}
w_n \geq w_{n_1} + w_{n_2} - wi_{n} & $We replace $w_{n_1}, w_{n_2}\\
w_n \geq (\frac{w_n}{2} + inc_{n_2} + wi_n + inc_{n_1}) + (\frac{w_n}{2} + inc_{n_2}) - wi_n & $we simplify$\\
w_n \geq \frac{w_n}{2} + inc_{n_2} + inc_{n_1} + \frac{w_n}{2} + inc_{n_2}\\
w_n \geq \frac{w_n}{2} + \frac{w_n}{2} + 2 * inc_{n_2} + inc_{n_1}\\
w_n \geq w_n + 2 * inc_{n_2} + inc_{n_1}\\
0 \geq 2 * inc_{n_2} + inc_{n_1}\\
\end{array}
$
}

But, this is a contradiction with $inc_{n_2} > 0$. Hence, $n_2 \gg n_1$.
\end{proof}

In case 4, the following lemma ensures that given two nodes whose weight is smaller than $\frac{w_n}{2}$ then the heaviest node is better.

\begin{lemma}
\label{lem_caminob}
Given a uniform MET $T = (N, E, M)$ whose root is $n \in N$, and two nodes $n_1, n_2 \in Sea(T)$, where $\frac{w_n}{2} \geq w_{n_1} > w_{n_2}$ then $n_1 \gg n_2$.
\end{lemma}

\begin{proof}
We can assume that $w_{n_1} = \frac{w_n}{2} - dec_{n_1}$ and $w_{n_2} = \frac{w_n}{2} - dec_{n_2}$ where $dec_{n_2} > dec_{n_1} \geq 0$.
Moreover, we know that $w_{n_1} + w_{n_2}=\frac{w_n}{2} - dec_{n_1} + \frac{w_n}{2} - dec_{n_2}$ and thus $w_{n_1} + w_{n_2}=w_n - dec_{n_1} - dec_{n_2}$. Therefore, because $dec_{n_2} > dec_{n_1} \geq 0$, we deduce that $w_n > w_{n_1} + w_{n_2}$. And as $w_{n_1} > w_{n_2}$ then, by Lemma~\ref{comparacion}, $n_1 \gg n_2$.
\end{proof}

If two nodes $n_1$ and $n_2$ are brothers and $n_1$ is better than $n_2$ then $n_1$ is better than any descendant of $n_2$. The following lemma proves this property that is complementary to 
Lemma~\ref{lem_caminoa} for case 1.

\begin{lemma}
\label{lem_caminoa2}
Given a uniform MET $T = (N, E, M)$ whose root is $n \in N$ and four nodes $n_1 \in N$ and $n_2, n_3, n_4 \in Sea(T)$ with $(n \rightarrow n_1) \in E^*$, $(n_1 \rightarrow n_2), (n_1 \rightarrow n_3) \in E,$ $(n_3 \rightarrow n_4)$ $\in E^+$, if $n_2 \gg n_3 \lor n_2 \equiv n_3$ then $n_2 \gg n_4$.
\end{lemma}

\begin{proof}
First, $n_2$ and $n_3$ are brothers and $n_2 \gg n_3 \lor n_2 \equiv n_3$ then, by Lemma~\ref{lem_caminoa}, we know that $w_{n_2} \geq w_{n_3}$. We distinguish two cases $w_{n_2} > \frac{w_n}{2}$ and $\frac{w_n}{2} \geq w_{n_2}$.\\
If $\frac{w_n}{2} \geq w_{n_2}$ then $\frac{w_n}{2} \geq w_{n_3}$ and by Lemma~\ref{lem_caminob} $n_3 \gg n_4$.\\
If $w_{n_2} > \frac{w_n}{2}$ then we only have to demonstrate that $\frac{w_n}{2} > w_{n_3}$ and then (as before) by Lemma~\ref{lem_caminob} $n_3 \gg n_4$.\\
This can be easily proved having into account that $w_n \ge w_{n_2} + w_{n_3}$ because $n_2$ and $n_3$ are children of $n_1$ and $n_1$ is descendant of $n$, and that $w_{n_2} = \frac{w_n}{2} + inc_{n_2}$ with $inc_{n_2} > 0$.

{\footnotesize
$
\begin{array}{l@{~~~~}r}
w_n \geq w_{n_2} + w_{n_3} & $we replace $w_{n_2}\\
w_n \geq (\frac{w_n}{2} + inc_{n_2}) + w_{n_3}\\
w_n - \frac{w_n}{2} \geq inc_{n_2} + w_{n_3}\\
\frac{w_n}{2} \geq inc_{n_2} + w_{n_3} & $as $inc_{n_2} > 0\\
\frac{w_n}{2} > w_{n_3}\\
\end{array}
$\\
}

Therefore as $n_2 \gg n_3 \lor n_2 \equiv n_3$ and $n_3 \gg n_4$ then $n_2 \gg n_4$.
\end{proof}

The previous lemmas allow Algorithm~\ref{alg_camino} to find a path between the root node and an optimal node. 
The correctness of this algorithm is proved by the following theorem.\\

\noindent {\bf Theorem~\ref{theo_corr}.}
\emph{
Let $T = (N, E, M)$ be a uniform MET, then the execution of Algorithm~\ref{alg_camino} with $T$ as input always terminates producing as output a node $n \in Sea(T)$ such that $\nexists n' \in Sea(T) \mid n' \gg n$.
}

\begin{proof}
The finiteness of the algorithm is proved thanks to the following invariant: $w_{Candidate}$ strictly decreases in each iteration. Therefore, because $N$ is finite, $w_{Candidate}$ will eventually become smaller or equal to $\frac{w_n}{2}$ and the loop will terminate. 

The correctness can be proved showing that after any number of iterations the algorithm always finishes with an optimal node. 
We prove it by induction on the number of iterations performed.

\medskip
\noindent{\bf(Base Case)} In the base case, only one iteration is executed. 
If the condition in Line (5) is satisfied then the root is marked as undefined and it is trivially the optimal node. This optimal node is returned in Line (5).
Otherwise, Lines (4) and (6) select the heaviest child of the root, the loop terminates and Lines (9) or (10) return the optimal node. 

Note that the root node---when it is marked as $\mathit{Wrong}$---can only be selected in the first iteration. 
But even in this case, this node is never selected because the root node must have at least one child marked as $\mathit{Undefined}$. Thus Line (5) is not satisfied and Line (6) selects this node. 
If the condition of the loop is not satisfied, then Line (8) returns the roots' child.

\medskip
\noindent{\bf(Induction Hypothesis)} We assume as the induction hypothesis that after $i$ iterations, the algorithm has a candidate node $Best \in Sea(T)$ such that $\forall n' \in Sea(T), (Best \rightarrow n') \not\in E^*, Best \gg n'$.

\medskip
\noindent{\bf(Inductive Case)} We now prove that the iteration $i+1$ of the algorithm will select a new candidate node $Candidate$ such that $Candidate \gg Best$, or it will terminate selecting an optimal node.

Firstly, when the condition in Line (5) is satisfied $\mathit{Best}$ and $\mathit{Candidate}$ are the same node (say $n'$). According to the induction hypothesis, this node is better than any other of the nodes in the set $\{n'' \in Sea(T) | (n' \rightarrow n'') \not\in E^*\}$. Therefore, because $n'$ has no children, then it is an optimal node; and it is returned in Line (5).
Otherwise, if the condition in Line (5) is not satisfied, Line (7) in the algorithm ensures that $w_{Best}>\frac{w_n}{2}$ being $n$ the root of $T$ because in the iteration $i$ the loop did not terminate or because $Best$ is the root.
Moreover, according to Lines (4) and (6), we know that $Candidate$ is the heaviest child of $Best$. 
We have two possibilities:
{\small
\begin{itemize}
\item $w_{\mathit{Candidate}} >\frac{w_n}{2}$: In this case the loop does not terminate and $\forall n' \in Sea(T),$ $(\mathit{Candidate} \rightarrow n') \not\in E^*, \mathit{Candidate} \gg n'$.
Firstly, by Lemma~\ref{lem_caminoc} we know that $\mathit{Candidate} \gg \mathit{Best}$, and thus, by the induction hypothesis we know that $\forall n' \in Sea(T), (\mathit{Best} \rightarrow n') \not\in E^*, \mathit{Candidate} \gg n'$. By Lemma~\ref{lem_caminoa} $\mathit{Candidate} \gg n'$ $\lor$ $\mathit{Candidate} \equiv n'$ being $n'$ a brother of $\mathit{Candidate}$. 
But as we know that $w_{\mathit{Candidate}} >\frac{w_n}{2}$ then $\mathit{Candidate} \not\equiv n'$.
Moreover, by Lemma~\ref{lem_caminoa2} we can ensure that $\mathit{Candidate} \gg n'$ being $n'$ a descendant of a $\mathit{candidate}$'s brother.
\item $w_{\mathit{Candidate}} \leq\frac{w_n}{2}$: In this case the loop terminates (Line (7)) and by Lemma~\ref{lem_caminoa} we know that $\mathit{Candidate} \gg n' \lor \mathit{Candidate} \equiv n'$ being $n'$ a brother of $\mathit{Candidate}$. Moreover, by Lemma~\ref{lem_caminoa2} we can ensure that $\mathit{Candidate} \gg n'$ being $n'$ a descendant of a $\mathit{candidate}$'s brother. Then equation $(w_n \geq w_{\mathit{Best}} + w_{\mathit{Candidate}} - wi_n)$ is applied in Line (9) to select an optimal node. Theorems~\ref{equation} and \ref{prop} ensures that the node selected is an optimal node because, according to Lemma~\ref{lem_caminob}, for all descendant $n'$ of $\mathit{Candidate}$, $\mathit{Candidate} \gg n'$.
\end{itemize}
}
\end{proof}

\newpage
\subsection{Proof of Theorem~\ref{theo_corr2}}


Theorem~\ref{theo_corr2} states the correctness of Algorithm~\ref{alg_general} used in the general case when nodes can have different individual weights. 
For the proof of this theorem we define first some auxiliary lemmas. 
The following lemma ensures that $w_{n_1} - \frac{wi_{n_1}}{2} > \frac{w_{n}}{2}$ used in the condition of the loop implies $d_{n_1} > u_{n_1}$.
\begin{lemma}
\label{lem_condicionBucle}
Given a variable MET $T = (N, E, M)$ whose root is $n \in N$ and a node $n_1 \in Sea(T)$, $d_{n_1} > u_{n_1}$ if and only if $w_{n_1} - \frac{wi_{n_1}}{2} > \frac{w_{n}}{2}$.
\end{lemma}

\begin{proof}
We proof that $w_{n_1} - \frac{wi_{n_1}}{2} > \frac{w_{n}}{2}$ implies $d_{n_1} > u_{n_1}$ and vice versa.

{\footnotesize
$
\begin{array}{l@{}r}
w_{n_1} - \frac{wi_{n_1}}{2} > \frac{w_{n}}{2}\\
2w_{n_1} - wi_{n_1} > w_{n}\\
$We replace $w_{n_1}$ using Equation 2:$\\
2(d_{n_1} + wi_{n_1}) - wi_{n_1} > w_{n}\\
2d_{n_1} + wi_{n_1} > w_{n}\\
d_{n_1} > w_{n} - d_{n_1} - wi_{n_1}\\
$We replace $w_{n} - d_{n_1} - wi_{n_1}$ using Equation 1:$\\
d_{n_1} > u_{n_1}\\
\end{array}
$\\
}

\end{proof}

The following lemma ensures that given two nodes $n_1$ and $n_2$ where $d_n \geq u_n$ in both nodes and $n_1 \rightarrow n_2$ then $n_2 \gg n_1 \lor n_2 \equiv n_1$.
\begin{lemma}
\label{lem_camino2c}
Given a variable MET $T = (N, E, M)$ and given two nodes $n_1, n_2 \in Sea(T)$, with $(n_1 \rightarrow n_2) \in E$, if $d_{n_2} \geq u_{n_2}$ then $n_2 \gg n_1 \lor n_2 \equiv n_1$.
\end{lemma}

\begin{proof}
We prove that $|d_{n_2} - u_{n_2}| \leq |d_{n_1} - u_{n_1}|$ holds.
First, we know that $d_{n_1} = d_{n_2} + wi_{n_2} + inc$ and $u_{n_1} = u_{n_2} - wi_{n_1} - inc$ with $inc \geq 0$, where $inc$ represent the weight of the possible brothers of $n_2$.

{\footnotesize
$
\begin{array}{l@{}r}
|d_{n_2} - u_{n_2}| \leq |d_{n_1} - u_{n_1}|\\
$As we know that $d_{n} \geq u_{n}$ in both nodes:$\\
d_{n_2} - u_{n_2} \leq d_{n_1} - u_{n_1}\\
$We replace $d_{n_1}$ and $u_{n_1}$:$\\
d_{n_2} - u_{n_2} \leq (d_{n_2} + wi_{n_2} + inc) - (u_{n_2} - wi_{n_1} - inc)\\
d_{n_2} - u_{n_2} \leq d_{n_2} - u_{n_2} + wi_{n_1} + wi_{n_2} + 2inc\\
0 \leq wi_{n_1} + wi_{n_2} + 2inc\\
\end{array}
$\\
}
Hence, because $wi_{n_1}$, $wi_{n_2}$, $inc \geq 0$ then $|d_{n_2} - u_{n_2}| \leq |d_{n_1} - u_{n_1}|$ is satisfied and thus $n_2 \gg n_1 \lor n_2 \equiv n_1$.
\end{proof}

The following lemma ensures that given two nodes $n_1$ and $n_2$ where $d_n \leq u_n$ in both nodes and $n_1 \rightarrow n_2$ then $n_1 \gg n_2 \lor n_1 \equiv n_2$.
\begin{lemma}
\label{lem_debajo}
Given a variable MET $T = (N, E, M)$ and given two nodes $n_1, n_2 \in Sea(T)$, with $(n_1 \rightarrow n_2) \in E$, if $d_{n_1} \leq u_{n_1}$ then $n_1 \gg n_2 \lor n_1 \equiv n_2$.
\end{lemma}

\begin{proof}
We prove that $|d_{n_1} - u_{n_1}| \leq |d_{n_2} - u_{n_2}|$ holds.
First, we know that $d_{n_2} = d_{n_1} - wi_{n_2} - inc$ and $u_{n_2} = u_{n_1} + wi_{n_1} + inc$ with $inc \geq 0$, where $inc$ represent the weight of the possible brothers of $n_2$.

{\footnotesize
$
\begin{array}{l@{}r}
|d_{n_1} - u_{n_1}| \leq |d_{n_2} - u_{n_2}|\\
$As we know that $u_{n} \geq d_{n}$ in both nodes:$\\
u_{n_1} - d_{n_1} \leq u_{n_2} - d_{n_2}\\
$We replace $d_{n_2}$ and $u_{n_2}$:$\\
u_{n_1} - d_{n_1} \leq (u_{n_1} + wi_{n_1} + inc) - (d_{n_1} - wi_{n_2} - inc)\\
u_{n_1} - d_{n_1} \leq u_{n_1} - d_{n_1} + wi_{n_1} + wi_{n_2} + 2inc\\
0 \leq wi_{n_1} + wi_{n_2} + 2inc\\
\end{array}
$\\
}
Hence, because $wi_{n_1}$, $wi_{n_2}$, $inc \geq 0$ then $|d_{n_1} - u_{n_1}| \leq |d_{n_2} - u_{n_2}|$ is satisfied and thus $n_1 \gg n_2 \lor n_1 \equiv n_2$.
\end{proof}

The following lemma ensures that given two brother nodes $n_1$ and $n_2$, if $d_{n_1} \geq u_{n_1}$ then $d_{n_2} \leq u_{n_2}$.
\begin{lemma}
\label{lem_contradiccionConD}
Given a variable MET $T = (N, E, M)$ whose root is $n \in N$, and given three nodes $n_1 \in N$ and $n_2, n_3 \in Sea(T)$, with $(n \rightarrow n_1) \in E^*$, $(n_1 \rightarrow n_2), (n_1 \rightarrow n_3) \in E$, if $d_{n_2} \geq u_{n_2}$ then $d_{n_3} \leq u_{n_3}$.
\end{lemma}
\begin{proof}
We prove it by contradiction assuming that $d_{n_3} > u_{n_3}$ when $d_{n_2} \geq u_{n_2}$ and they are brothers.
First, we know that as $n_2$ and $n_3$ are brothers then $u_{n_2} \geq w_{n_3}$ and $u_{n_3} \geq w_{n_2}$. Therefore, if $d_{n_3} > u_{n_3}$ then $d_{n_2} \geq u_{n_2} \geq w_{n_3} \geq d_{n_3} > u_{n_3} \geq w_{n_2} \geq d_{n_2}$ that implies $d_{n_2} > d_{n_2}$ that is a contradiction itself.
\end{proof}

If two nodes $n_1$ and $n_2$ are brothers and $d_{n_1} \geq u_{n_1}$ then $n_1 \gg n_2 \lor n_1 \equiv n_2$. The following lemma proves this property.
\begin{lemma}
\label{lem_hermanos}
Given a variable MET $T = (N, E, M)$ whose root is $n \in N$, and given three nodes $n_1 \in N$ and $n_2, n_3 \in Sea(T)$, with $(n \rightarrow n_1) \in E^*$, $(n_1 \rightarrow n_2), (n_1 \rightarrow n_3) \in E$, if $d_{n_2} \geq u_{n_2}$ then $n_2 \gg n_3 \lor n_2 \equiv n_3$.
\end{lemma}

\begin{proof}
We prove that $|d_{n_2} - u_{n_2}| \leq |d_{n_3} - u_{n_3}|$ holds.
First, as $n_2$ and $n_3$ are brothers we know that $w_n \geq d_{n_2} + d_{n_3} + wi_{n_2} + wi_{n_3}$, then $w_n = d_{n_2} + d_{n_3} + wi_{n_2} + wi_{n_3} + inc$ with $inc \geq 0$.

{\footnotesize
$
\begin{array}{l@{}r}
|d_{n_2} - u_{n_2}| \leq |d_{n_3} - u_{n_3}|\\
$As $d_{n_2} \geq u_{n_2}$ by Lemma~\ref{lem_contradiccionConD} we know that $u_{n_3} \geq d_{n_3}$:$\\
d_{n_2} - u_{n_2} \leq u_{n_3} - d_{n_3}\\
$We replace $u_{n_2}$ and $u_{n_3}$ using Equation 1:$\\
d_{n_2} - (w_n - d_{n_2} - wi_{n_2}) \leq (w_n - d_{n_3} - wi_{n_3}) - d_{n_3}\\
-w_n + 2d_{n_2} + wi_{n_2} \leq w_n - 2d_{n_3} - wi_{n_3}\\
-2w_n \leq -2d_{n_2} - 2d_{n_3} - wi_{n_2} - wi_{n_3}\\
2w_n \geq 2d_{n_2} + 2d_{n_3} + wi_{n_2} + wi_{n_3}\\
w_n \geq d_{n_2} + d_{n_3} + \frac{wi_{n_2}}{2} + \frac{wi_{n_3}}{2}\\
$We replace $w_n$:$\\
d_{n_2} + d_{n_3} + wi_{n_2} + wi_{n_3} + inc \geq d_{n_2} + d_{n_3} + \frac{wi_{n_2}}{2} + \frac{wi_{n_3}}{2}\\
wi_{n_2} + wi_{n_3} + inc \geq \frac{wi_{n_2}}{2} + \frac{wi_{n_3}}{2}\\
\frac{wi_{n_2}}{2} + \frac{wi_{n_3}}{2} + inc \geq 0\\
\end{array}
$\\
}
Hence, because $wi_{n_2}$, $wi_{n_3}$, $inc \geq 0$ then $|d_{n_2} - u_{n_2}| \leq |d_{n_3} - u_{n_3}|$ is satisfied and thus $n_2 \gg n_3 \lor n_2 \equiv n_3$.
\end{proof}

The following lemma ensures that given two brother nodes $n_1$ and $n_2$, if $w_{n_1} \geq w_{n_2}$ and $d_{n_1} \leq u_{n_1}$ then $d_{n_2} \leq u_{n_2}$.
\begin{lemma}
\label{lem_contradiccionConD2}
Given a variable MET $T = (N, E, M)$ whose root is $n \in N$, and given three nodes $n_1 \in N$ and $n_2, n_3 \in Sea(T)$, with $(n \rightarrow n_1) \in E^*$, $(n_1 \rightarrow n_2), (n_1 \rightarrow n_3) \in E$, if $w_{n_2} \geq w_{n_3}$ and $d_{n_2} \leq u_{n_2}$ then $d_{n_3} \leq u_{n_3}$.
\end{lemma}
\begin{proof}
We prove it by contradiction assuming that $d_{n_3} > u_{n_3}$ when $w_{n_2} \geq w_{n_3}$ and $d_{n_2} \leq u_{n_2}$ and they are brothers.
First, we know that as $n_2$ and $n_3$ are brothers then $u_{n_2} \geq w_{n_3}$ and $u_{n_3} \geq w_{n_2}$. Therefore, if $d_{n_3} > u_{n_3}$ then $d_{n_3} > u_{n_3} \geq w_{n_2} \geq w_{n_3} \geq d_{n_3}$ that implies $d_{n_3} > d_{n_3}$ that is a contradiction itself.
\end{proof}

If two nodes $n_1$ and $n_2$ are brothers and $u_{n_1} \geq d_{n_1} \land u_{n_2} \geq d_{n_2}$ then, if $w_{n_1} - \frac{wi_{n_1}}{2} \geq w_{n_2} - \frac{wi_{n_2}}{2}$ is satisfied then $n_1 \gg n_2 \lor n_1 \equiv n_2$. The following lemma proves this property.
\begin{lemma}
\label{lem_hermanosDebajo}
Given a variable MET $T = (N, E, M)$ whose root is $n \in N$, and given three nodes $n_1 \in N$ and $n_2, n_3 \in Sea(T)$, with $(n \rightarrow n_1) \in E^*$, $(n_1 \rightarrow n_2), (n_1 \rightarrow n_3) \in E$, and $u_{n_2} \geq d_{n_2}$ and $u_{n_3} \geq d_{n_3}$, 
$n_2 \gg n_3 \lor n_2 \equiv n_3$ if and only if $w_{n_2} - \frac{wi_{n_2}}{2} \geq w_{n_3} - \frac{wi_{n_3}}{2}$.
\end{lemma}

\begin{proof}
First, if $|d_{n_2} - u_{n_2}| \leq |d_{n_3} - u_{n_3}|$ then $n_2 \gg n_3 \lor n_2 \equiv n_3$.
Thus it is enough to prove that $w_{n_2} - \frac{wi_{n_2}}{2} \geq w_{n_3} - \frac{wi_{n_3}}{2}$ implies $|d_{n_2} - u_{n_2}| \leq |d_{n_3} - u_{n_3}|$ and vice versa when $u_{n} \geq d_{n}$ in both nodes and they are brothers.

{\footnotesize
$
\begin{array}{l@{}r}
w_{n_2} - \frac{wi_{n_2}}{2} \geq w_{n_3} - \frac{wi_{n_3}}{2}\\
2w_{n_2} - wi_{n_2} \geq 2w_{n_3} - wi_{n_3}\\
$We replace $w_{n_2}$ and $w_{n_3}$ using Equation 2:$\\
2(d_{n_2} + wi_{n_2}) - wi_{n_2} \geq 2(d_{n_3} + wi_{n_3}) - wi_{n_3}\\
2d_{n_2} + wi_{n_2} \geq 2d_{n_3} + wi_{n_3}\\
$We add $-w_{n}$:$\\
-w_{n} + 2d_{n_2} + wi_{n_2} \geq -w_{n} + 2d_{n_3} + wi_{n_3}\\
w_{n} - 2d_{n_2} - wi_{n_2} \leq w_{n} - 2d_{n_3} - wi_{n_3}\\
$We replace $w_{n}$ using Equation 1:$\\
(d_{n_2} + u_{n_2} + wi_{n_2}) - 2d_{n_2} - wi_{n_2} \leq (d_{n_3} + u_{n_3} + wi_{n_3}) - 2d_{n_3} - wi_{n_3}\\
-d_{n_2} + u_{n_2} \leq -d_{n_3} + u_{n_3}\\
u_{n_2} - d_{n_2} \leq u_{n_3} - d_{n_3}\\
$As $u_{n} \geq d_{n}$ in both nodes:$\\
|u_{n_2} - d_{n_2}| \leq |u_{n_3} - d_{n_3}|\\
|d_{n_2} - u_{n_2}| \leq |d_{n_3} - u_{n_3}|\\
\end{array}
$\\
}
\end{proof}

If two nodes $n_1$ and $n_2$ are brothers and $d_{n_1} \geq u_{n_1}$ and $n_2 \rightarrow^+ n_3$ then, if $n_1 \equiv n_2$ then $n_1 \gg n_3 \lor n_1 \equiv n_3$. The following lemma proves this property.
\begin{lemma}
\label{lem_descendienteHermano}
Given a variable MET $T = (N, E, M)$ whose root is $n \in N$, and given four nodes $n_1 \in N$ and $n_2, n_3, n_4 \in Sea(T)$, with $(n \rightarrow n_1) \in E^*$, $(n_1 \rightarrow n_2), (n_1 \rightarrow n_3) \in E$, $(n_3 \rightarrow n_4) \in E^+$, if $d_{n_2} \geq u_{n_2}$ and $n_2 \equiv n_3$ then $n_2 \gg n_4 \lor n_2 \equiv n_4$.
\end{lemma}

\begin{proof}
This can be trivially proof having into account that $d_{n_3} \leq u_{n_3}$ when $d_{n_2} \geq u_{n_2}$ by Lemma~\ref{lem_contradiccionConD} and then by Lemma~\ref{lem_debajo} we know that $n_3 \gg n_4 \lor n_3 \equiv n_4$ and as $n_2 \equiv n_3$ then $n_2 \gg n_4 \lor n_2 \equiv n_4$.
\end{proof}

If two nodes $n_1$ and $n_2$ are brothers and $d_{n_1} \leq u_{n_1} \land d_{n_2} \leq u_{n_2}$ and $n_2 \rightarrow^+ n_3$ then, if $n_1 \equiv n_2$ then $n_1 \gg n_3 \lor n_1 \equiv n_3$. The following lemma proves this property.
\begin{lemma}
\label{lem_descendienteHermano2}
Given a variable MET $T = (N, E, M)$ whose root is $n \in N$, and given four nodes $n_1 \in N$ and $n_2, n_3, n_4 \in Sea(T)$, with $(n \rightarrow n_1) \in E^*$, $(n_1 \rightarrow n_2), (n_1 \rightarrow n_3) \in E$, $(n_3 \rightarrow n_4) \in E^+$, if $d_{n_2} \leq u_{n_2}$ and $d_{n_3} \leq u_{n_3}$ and $n_2 \equiv n_3$ then $n_2 \gg n_4 \lor n_2 \equiv n_4$.
\end{lemma}

\begin{proof}
This can be trivially proof having into account that $d_{n_3} \leq u_{n_3}$ and then by Lemma~\ref{lem_debajo} we know that $n_3 \gg n_4 \lor n_3 \equiv n_4$ and as $n_2 \equiv n_3$ then $n_2 \gg n_4 \lor n_2 \equiv n_4$.
\end{proof}

If two nodes $n_1$ and $n_2$ are brothers and $n_1 \gg n_2$ and $n_2 \rightarrow^+ n_3$ then $n_1 \gg n_3$. The following lemma proves this property.
\begin{lemma}
\label{lem_camino2a2}
Given a variable MET $T = (N, E, M)$ whose root is $n \in N$, and given four nodes $n_1 \in N$ and $n_2, n_3, n_4 \in Sea(T)$, with $(n \rightarrow n_1) \in E^*$, $(n_1 \rightarrow n_2), (n_1 \rightarrow n_3) \in E$, $(n_3 \rightarrow n_4) \in E^+$, if $n_2 \gg n_3$ then $n_2 \gg n_4$.
\end{lemma}

\begin{proof}
We show that if $n_2 \gg n_3$ then $d_{n_3} < u_{n_3}$. We prove it by contradiction assuming that $d_{n_3} \geq u_{n_3}$ when $n_2 \gg n_3$.
First, as $n_2$ and $n_3$ are brothers we know that $w_n \geq d_{n_2} + d_{n_3} + wi_{n_2} + wi_{n_3}$, then $w_n = d_{n_2} + d_{n_3} + wi_{n_2} + wi_{n_3} + inc$ with $inc \geq 0$. Therefore, if $|d_{n_2} - u_{n_2}| < |d_{n_3} - u_{n_3}|$ then $n_2 \gg n_3$.
Thus it is enough to prove that $|d_{n_2} - u_{n_2}| < |d_{n_3} - u_{n_3}|$ is not satisfied when $d_{n_3} \geq u_{n_3}$ and $n_2$ and $n_3$ are brothers.

{\footnotesize
$
\begin{array}{l@{}r}
|d_{n_2} - u_{n_2}| < |d_{n_3} - u_{n_3}|\\
$As $d_{n_3} \geq u_{n_3}$ by Lemma~\ref{lem_contradiccionConD} we know that $u_{n_2} \geq d_{n_2}$:$\\
u_{n_2} - d_{n_2} < d_{n_3} - u_{n_3}\\
$We replace $u_{n_2}$ and $u_{n_3}$ using Equation 1:$\\
(w_{n} - d_{n_2} - wi_{n_2}) - d_{n_2} < d_{n_3} - (w_{n} - d_{n_3} - wi_{n_3})\\
w_{n} - 2d_{n_2} - wi_{n_2} < 2d_{n_3} - w_{n} + wi_{n_3}\\
2w_{n} < 2d_{n_2} + 2d_{n_3} + wi_{n_2} + wi_{n_3}\\
w_{n} < d_{n_2} + d_{n_3} + \frac{wi_{n_2}}{2} + \frac{wi_{n_3}}{2}\\
$We replace $w_{n}$:$\\
d_{n_2} + d_{n_3} + wi_{n_2} + wi_{n_3} + inc < d_{n_2} + d_{n_3} + \frac{wi_{n_2}}{2} + \frac{wi_{n_3}}{2}\\
wi_{n_2} + wi_{n_3} + inc < \frac{wi_{n_2}}{2} + \frac{wi_{n_3}}{2}\\
\frac{wi_{n_2}}{2} + \frac{wi_{n_3}}{2} + inc < 0\\
\end{array}
$\\
}
But, this is a contradiction with $wi_{n_2}, wi_{n_3}, inc \geq 0$. Hence, $d_{n_3} < u_{n_3}$.

Now we show that, if $n_2 \gg n_3$ then $n_2 \gg n_4$. We prove it by contradiction assuming that $n_4 \gg n_2 \lor n_4 \equiv n_2$ when $n_2 \gg n_3$.
First, we know that $d_{n_3} < u_{n_3}$. Therefore we know that $d_{n_4} = d_{n_3} - wi_{n_4} - dec$ and $u_{n_4} = u_{n_3} + wi_{n_3} + dec$ with $dec \geq 0$, where $dec$ represent the weight of the possible brothers of $n_4$.

{\footnotesize
$
\begin{array}{l@{}r}
|d_{n_3} - u_{n_3}| > |d_{n_2} - u_{n_2}| \geq |d_{n_4} - u_{n_4}|\\
$We replace $d_{n_4}$ and $u_{n_4}$:$\\
|d_{n_3} - u_{n_3}| > |d_{n_2} - u_{n_2}| \geq |(d_{n_3} - wi_{n_4} - dec) - (u_{n_3} + wi_{n_3} + dec)|\\
|d_{n_3} - u_{n_3}| > |d_{n_2} - u_{n_2}| \geq |d_{n_3} - wi_{n_4} - dec - u_{n_3} - wi_{n_3} - dec|\\
|d_{n_3} - u_{n_3}| > |d_{n_2} - u_{n_2}| \geq |d_{n_3} - u_{n_3} - wi_{n_3} - wi_{n_4} - 2dec|\\
\end{array}
$\\
}
Note that $d_{n_3} - u_{n_3}$ must be positive, thus $d_{n_3} > u_{n_3}$. But this is a contradiction with $d_{n_3} < u_{n_3}$.
\end{proof}

The following lemma ensures that given two nodes $n_1$ and $n_2$ where $d_{n_1} \geq u_{n_1}$ and $d_{n_2} \leq u_{n_2}$ and $n_1 \rightarrow n_2$ then if $w_{n} \geq w_{n_1} + w_{n_2} - \frac{wi_{n_1}}{2} - \frac{wi_{n_2}}{2}$ is satisfied then $n_1 \gg n_2 \lor n_1 \equiv n_2$.
\begin{lemma}
\label{lem_camino2e}
Given a variable MET $T = (N, E, M)$ and given two nodes $n_1, n_2 \in Sea(T)$, with $(n_1 \rightarrow n_2) \in E$, and $d_{n_1} \geq u_{n_1}$, and $d_{n_2} \leq u_{n_2}$, 
$n_1 \gg n_2 \lor n_1 \equiv n_2$ if and only if $w_{n} \geq w_{n_1} + w_{n_2} - \frac{wi_{n_1}}{2} - \frac{wi_{n_2}}{2}$.
\end{lemma}

\begin{proof}
First, if $|d_{n_1} - u_{n_1}| \leq |d_{n_2} - u_{n_2}|$ then $n_1 \gg n_2$ or $n_1 \equiv n_2$.
Thus it is enough to prove that $w_{n} \geq w_{n_1} + w_{n_2} - \frac{wi_{n_1}}{2} - \frac{wi_{n_2}}{2}$ implies $|d_{n_1} - u_{n_1}| \leq |d_{n_2} - u_{n_2}|$ and vice versa when $d_{n_1} \geq u_{n_1}$ and $d_{n_2} \leq u_{n_2}$.

{\footnotesize
$
\begin{array}{l@{}r}
w_{n} \geq w_{n_1} + w_{n_2} - \frac{wi_{n_1}}{2} - \frac{wi_{n_2}}{2}\\
$We replace $w_{n_1}, w_{n_2}$ using Equation 2:$\\
w_{n} \geq (d_{n_1} + wi_{n_1}) + (d_{n_2} + wi_{n_2}) - \frac{wi_{n_1}}{2} - \frac{wi_{n_2}}{2}\\
w_{n} \geq d_{n_1} + d_{n_2} + \frac{wi_{n_1}}{2} + \frac{wi_{n_2}}{2}\\
2w_{n} \geq 2d_{n_1} + 2d_{n_2} + wi_{n_1} + wi_{n_2}\\
-2w_{n} \leq -2d_{n_1} - 2d_{n_2} - wi_{n_1} - wi_{n_2}\\
-w_{n} + 2d_{n_1} + wi_{n_1} \leq w_{n} - 2d_{n_2} - wi_{n_2}\\
$We replace $w_{n}$ using Equation 1:$\\
-(d_{n_1} + u_{n_1} + wi_{n_1}) + 2d_{n_1} + wi_{n_1} \leq (d_{n_2} + u_{n_2} + wi_{n_2}) - 2d_{n_2} - wi_{n_2}\\
-d_{n_1} - u_{n_1} - wi_{n_1} + 2d_{n_1} + wi_{n_1} \leq d_{n_2} + u_{n_2} + wi_{n_2} - 2d_{n_2} - wi_{n_2}\\
-u_{n_1} + d_{n_1} \leq -d_{n_2} + u_{n_2}\\
d_{n_1} - u_{n_1} \leq u_{n_2} - d_{n_2}\\
$As $d_{n_1} \geq u_{n_1}$ and $d_{n_2} \leq u_{n_2}$:$\\
|d_{n_1} - u_{n_1}| \leq |u_{n_2} - d_{n_2}|\\
|d_{n_1} - u_{n_1}| \leq |d_{n_2} - u_{n_2}|\\
\end{array}
$\\
}
\end{proof}

Finally, we prove the correctness of Algorithm~\ref{alg_general}.\\

\noindent {\bf Theorem~\ref{theo_corr2}.}
\emph{
Let $T = (N, E, M)$ be a variable MET, then the execution of Algorithm~\ref{alg_general} with $T$ as input always terminates producing as output a node $n \in Sea(T)$ such that $\nexists n' \in Sea(T) \mid n' \gg n$.
}

\begin{proof}
The finiteness of the algorithm is proved thanks to the following invariant: each iteration processes one single node, and the same node is never processed again. Therefore, because $N$ is finite, the loop will terminate. 

The proof of correctness is completely analogous to the proof of Theorem~\ref{theo_corr}. The only difference is the induction hypothesis and the inductive case:

\medskip
\noindent {\bf(Induction Hypothesis)} After $i$ iterations, the algorithm has a candidate node $Best \in Sea(T)$ such that $\forall n' \in Sea(T), (Best \rightarrow n') \not\in E^*, Best \gg n' \vee Best \equiv n'$.

\medskip
\noindent {\bf(Inductive Case)} We prove that the iteration $i+1$ of the algorithm will select a new candidate node $Candidate$ such that $Candidate \gg Best \lor Candidate \equiv Best$, or it will terminate selecting an optimal node.\\
Firstly, when the condition in Line (5) is satisfied $\mathit{Best}$ and $\mathit{Candidate}$ are the same node (say $n'$). According to the induction hypothesis, this node is better or equal than any other of the nodes in the set $\{n'' \in Sea(T) | (n' \rightarrow n'') \not\in E^*\}$. Therefore, because $n'$ has no children, then it is an optimal node; and it is returned in Line (5).
Otherwise, if the condition in Line (5) is not satisfied, Line (7) in the algorithm ensures that $w_{\mathit{Best}} - \frac{wi_{\mathit{Best}}}{2} > \frac{w_{n}}{2}$ being $n$ the root of $T$ because in the iteration $i$ the loop did not terminate or because $Best$ is the root (observe that an exception can happen when all nodes have an individual weight of 0. But in this case all nodes are optimal, and thus the node returned by the algorithm is optimal). Then we know that $d_{\mathit{Best}} > u_{\mathit{Best}}$ by Lemma~\ref{lem_condicionBucle}.
Moreover, according to Lines (4) and (6), we know that $Candidate$ is the heaviest child of $Best$.
We have two possibilities:
{\small
\begin{itemize}
\item $d_{\mathit{Candidate}} > u_{\mathit{Candidate}}$: In this case the loop does not terminate and $\forall n' \in Sea(T),$ $(\mathit{Candidate} \rightarrow n') \not\in E^*, \mathit{Candidate} \gg n' \vee \mathit{Candidate} \equiv n'$.
Firstly, by Lemma~\ref{lem_camino2c} we know that $\mathit{Candidate} \gg \mathit{Best} \vee \mathit{Candidate} \equiv \mathit{Best}$, and thus, by the induction hypothesis we know that $\forall n' \in Sea(T), (\mathit{Best} \rightarrow n') \not\in E^*, \mathit{Candidate} \gg n' \vee \mathit{Candidate} \equiv n'$. By Lemma~\ref{lem_hermanos} we know that $\mathit{Candidate} \gg n' \vee \mathit{Candidate} \equiv n'$ being $n'$ a brother of $\mathit{Candidate}$. Moreover, by Lemma~\ref{lem_descendienteHermano} and \ref{lem_camino2a2} we can ensure that $\mathit{Candidate} \gg n' \vee \mathit{Candidate} \equiv n'$ being $n'$ a descendant of a $\mathit{candidate}$'s brother.\\
\item $d_{\mathit{Candidate}} \leq u_{\mathit{Candidate}}$: In this case the loop terminates (Line (7)) and we know by Lemma~\ref{lem_contradiccionConD2} that $d_{n'} \leq u_{n'}$ being $n'$ any brother of $\mathit{Candidate}$. In Line (8) according to Lemma~\ref{lem_hermanosDebajo} we select the Candidate such that $\mathit{Candidate} \gg n' \lor \mathit{Candidate} \equiv n'$ being $n'$ a brother of $\mathit{Candidate}$. Moreover, by Lemma~\ref{lem_descendienteHermano2} and \ref{lem_camino2a2} we can ensure that $\mathit{Candidate} \gg n' \vee \mathit{Candidate} \equiv n'$ being $n'$ a descendant of a $\mathit{candidate}$'s brother. Then equation $(w_n \geq w_{Best} + w_{Candidate} - \frac{wi_{Best}}{2} - \frac{wi_{Candidate}}{2})$ is applied in Line (10) to select an optimal node. Lemma~\ref{lem_camino2e} ensure that the node selected is an optimal node because, according to Lemma~\ref{lem_debajo}, for all descendant $n'$ of $Candidate$, $Candidate \gg n' \vee Candidate \equiv n'$.
\end{itemize}
}
\end{proof}

\end{document}